\documentclass{article}
\usepackage[numbers]{natbib}
\usepackage{fontenc}
\usepackage{shadethm}
\usepackage{amsthm}
\usepackage{float}
\usepackage[normalem]{ulem}
\usepackage{bm}
\usepackage{mathtools}
\usepackage{graphicx}
\usepackage{subfig}
\usepackage{stmaryrd}
\usepackage{sectsty}
\usepackage{mathrsfs} 
\usepackage{amsmath}
\usepackage{amssymb}
\usepackage{wasysym}
\usepackage[hyperfootnotes=false]{hyperref}
\usepackage[a4paper, left=2.7cm, right=2.7cm, top=2.5cm]{geometry}
\usepackage{chngcntr}
\usepackage{cleveref}
\usepackage[utf8]{inputenc}
\usepackage[german,russian,english]{babel}
\usepackage{unnumberedtotoc}

\counterwithin*{equation}{section}

\DeclareMathAlphabet{\mathpzc}{OT1}{pzc}{m}{it}

\sectionfont{\normalsize\centering}
\subsectionfont{\footnotesize\centering}

\theoremstyle{plain}
\newtheorem{thm}{Theorem}[section] 

\theoremstyle{definition}
\newtheorem{defn}[thm]{Definition} 
\newtheorem{lem}[thm]{Lemma}

\newtheorem{rem}[thm]{Remark}
\newtheorem{cor}[thm]{Corollary}

\def\XXint#1#2#3{{\setbox0=\hbox{$#1{#2#3}{\int}$ }
		\vcenter{\hbox{$#2#3$ }}\kern-.6\wd0}}

\newcommand\scalemath[2]{\scalebox{#1}{\mbox{\ensuremath{\displaystyle #2}}}}

\newcounter{MPequ}

\newcounter{AppA}

\newcounter{AppB}

\newcounter{AppC}

\newcounter{AppD}

\newcounter{AppE}

\begin{document}\selectlanguage{english}
\begin{center}
\normalsize \textbf{{Electric current dynamics in the stellarator coil winding surface model}}
\end{center}

\begin{center}
	Wadim Gerner$^{\dagger,}$\footnote{\textit{E-mail address:} \href{mailto:wadim.gerner@edu.unige.it}{wadim.gerner@edu.unige.it}, corresponding author}, Anouk Nicolopoulos-Salle$^{\star,}$\footnote{\textit{E-mail address:} \href{mailto:anouk.nicolopoulos@renfusion.eu}{anouk.nicolopoulos@renfusion.eu}}, Diego Pereira Botelho$^{\star,}$\footnote{\textit{E-mail address:} \href{mailto:diego.pereira@renfusion.eu}{diego.pereira@renfusion.eu}}
\end{center}

\begin{center}
{\footnotesize MaLGa Center, Department of Mathematics, Department of Excellence 2023-2027, University of Genoa, Via Dodecaneso 35, 16146 Genova, Italy$^{\dagger}$}
\end{center}
\begin{center}
{\footnotesize Renaissance Fusion, 38600 Fontaine, France$^{\star}$}
\end{center}
{\small \textbf{Abstract:} 
In stellarator design, the coil winding surfaces $\Sigma\subset\mathbb R^3$ support
current distributions $j$ that shape the magnetic field.
This work provides a theoretical framework explaining the emergence of
centre and saddle point regions, a key feature in coil optimisation.
For coil winding surfaces with a toroidal shape, we prove a dichotomy principle:
the current distribution has both centre and saddle point regions or is no-where vanishing.
For coil winding surfaces that consist of piecewise cylinders, 
we show that if $j$ is oppositely oriented on the two boundary circles,
centre and saddle points appear, and all but finitely many field lines of $j$ are periodic.
When $j$ admits a harmonic potential, all field lines are closed poloidal orbits.
These results offer insights into current patterns on winding surfaces,
with implications for coil design strategies and their simplification.

\vspace{1cm}

\section{Introduction}
\label{S1}

\subsection{Motivation}
\label{Sub11}

Nuclear fusion based on magnetic confinement of plasma relies primarily on two main approaches. The first is the tokamak, which confines plasma using a helical magnetic field generated in part by a toroidal current driven through the plasma itself. The second is the stellarator, which produces the confining magnetic field entirely through external non-planar coils, eliminating the need for a plasma current. The main advantage of the tokamak lies in its relatively simple coil geometry, while its principal disadvantage is its reliance on a plasma current that must be actively driven and sustained and which can give rise to current-driven instabilities and disruptions. Conversely, the stellarator operates inherently in steady state and avoids such instabilities, but its complex three-dimensional coil geometry poses significant engineering and manufacturing challenges. For a general overview, the reader is referred to \cite{Xu16},\cite{Hel14},\cite{Hel18}.
Although the stellarator concept predates the tokamak, it was set aside in the late 1960s due to its design complexity and the promising results observed in early tokamak experiments.
The international fusion research programme ITER, which has driven and continues to drive numerous research advancements, adopted the tokamak approach.
Today, however, the challenges faced by tokamaks and the successes of recent stellarator experiments have renewed interest in the latter.
Notably, the stellarator has emerged as the leading approach among private fusion companies, with one of the primary engineering challenges being the optimisation of coil geometry to achieve the simplest possible design while maintaining the required confinement properties.
\newline
\newline
Obtaining suitable plasma configurations and corresponding supporting magnetic fields
is far from simple and is the subject of ongoing research \cite{APL19,PALC20,SPWB21,JGCHJL23}.
Traditionally, stellarator design and optimisation consists of two stages:
\begin{enumerate}
	\item Equilibrium optimisation: the plasma domain and corresponding magnetic field configuration are optimised taking into account confinement and stability properties, as well as other physical and engineering constraints. Software packages such as DESC or VMEC can be used in this context~\cite{HW83,DK20,PCDUK23Part1,CDPK23Part2,DCPK23Part3,DCPKUK24}.
	\item Coil optimisation: starting from the result of the previous stage, one aims to design a set of coils that reproduces the desired target magnetic field configuration in the plasma domain $\mathcal{P}$ with acceptable accuracy~\cite{L17,PLBD18,PRS22,RV22}.
\end{enumerate}
Recently, substantial efforts have been directed towards the development of a single-stage optimisation procedure \cite{HHPH21},\cite{WGLCS22},\cite{JGLRW23},\cite{JGL24}, \cite{SLBB24}.
In general, coil optimisation implies minimising, but not eliminating, the error with respect to a target magnetic field. Although this residual error may be small, it can significantly influence the plasma confinement properties. Consequently, jointly optimising the coil geometry and the plasma configuration may yield better results than optimising the coils for a fixed plasma configuration.
\newline
One way to proceed is to start from a coil  winding surface \cite{M87,IGPW24} which reproduces a
current distribution on a surface $\Sigma\subset\mathbb{R}^3$.
To set up the problem of interest we consider a region $\mathcal{P}$
within which we wish to confine the plasma, hereafter called plasma region, and a target magnetic field 
$B_T\in L^2(\mathcal{P},\mathbb{R}^3)$ which we wish to generate
by means of our coil currents. 
Since, according to Maxwell's equations, the coil current induced
magnetic field is divergence free as well as curl-free within $\mathcal{P}$,
we will be considering only divergence free and curl-free target magnetic fields.

We define the full field error
\begin{gather}
	\label{chi_B}
	\chi_B:=\|B(j)-B_T\|_{L^2(\mathcal{P},\mathbb{R}^3)},
\end{gather}
where $B(j)$ denotes the magnetic field induced by a current $j$.
Minimising this quantity results in an ill-posed problem, which is why a Tikhonov regularisation term is introduced to determine an optimal current distribution.
A common choice for the regularisation term is
\begin{gather}
	\label{chi_j}
	\chi_j:=\|j\|_{L^2(\Sigma)},
\end{gather}
which, for a fixed number of contours, encourages their distribution across the domain.
This approach promotes coil spacing and facilitates access for heating and diagnostics.
The optimisation of the current is then formulated as 
\begin{gather}
	\label{1E1}
	\min_{j}\{\chi_B^2+\lambda \chi_j^2\}.
\end{gather}
The minimisation is carried out over all square-integrable vector fields $j$ which are tangent to $\Sigma$
and div-free as vector fields on $\Sigma$.
Such vector fields $j$ are referred to as currents or current distributions.
For a given $\lambda>0$, the second term in \eqref{1E1} favors solutions with lower values of $\chi_j^2$.
This regularisation technique is widely used in the stellarator coil optimisation community to reduce coil complexity.
Other criteria, such as minimising coil curvature, length, or proximity to the plasma, may also be considered.
Within the stellarator coil optimisation community, a standard choice for the field error term is 
$\chi_B:=\|(B(j)-B_T)\cdot\mathcal{N}\|_{L^2(\partial\mathcal{P})}$,
where $\mathcal{N}$ is the outward unit normal to $\partial\mathcal{P}$.
This choice is motivated by the fact that the normal component of the magnetic field on the plasma boundary
uniquely determines the field throughout the plasma domain $\mathcal{P}$ for a fixed toroidal circulation, 
as detailed in \cite[Section 4]{PRS22}.
A detailed discussion on these aspects can be found in \cite[Section 13.4]{IGPW24}.
While the choice of a coil winding surface $\Sigma$ remains a degree of freedom and can be taken into account in the optimisation process \cite{PRS22},
we assume here that a suitable coil winding surface $\Sigma$ has already been determined.
\newline
We focus on the minimisation problem defined by (\ref{chi_B}-\ref{chi_j}-\ref{1E1}).
Standard variational arguments and the strict convexity of the functional ensure that, for any $\lambda>0$, this problem always admits a unique minimiser $j_{\lambda}$.
It was recently proven that $\|B(j_{\lambda})-B_T\|_{L^2(\mathcal{P},\mathbb{R}^3)}\rightarrow 0$ as $\lambda\searrow 0$, meaning that arbitrary precision can, at least theoretically, be achieved \cite[Section 4]{G24}.
In practice, a small parameter $\lambda>0$ is fixed, and the corresponding minimiser $j_{\lambda}$ is computed numerically.

Once the optimal current distribution on the coil winding surface is determined, it must be
approximated using coils. Coil shapes are derived from selected current streamlines and are practically realized by winding a conductive material over a frame.
Traditionally, stellarator coils are constructed from superconducting cables wound onto non-planar frames, producing a filamentary (one-dimensional) current flow.
However, recent technological advancements suggest the use of wide high-temperature superconductive (HTS) surfaces \cite[HTS patents]{RFPatents}.
These can be deposited or wound onto a cylinder, enabling a two-dimensional surface current.
In both cases, understanding the dynamics of surface electric currents is essential.

\begin{figure}[h]
	\centering
	\begin{subfloat}
		\centering
		\includegraphics[width=0.42\textwidth, keepaspectratio]{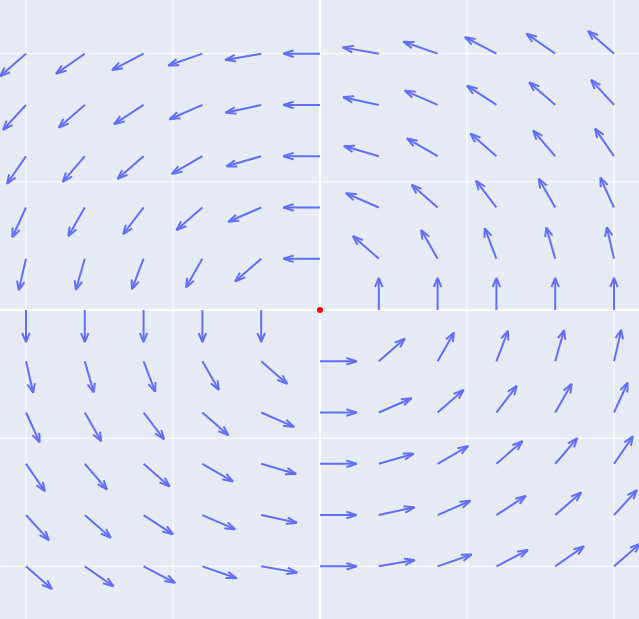}
	\end{subfloat}
	\hspace{2cm}
	\begin{subfloat}
		\centering
		\includegraphics[width=0.42\textwidth, keepaspectratio]{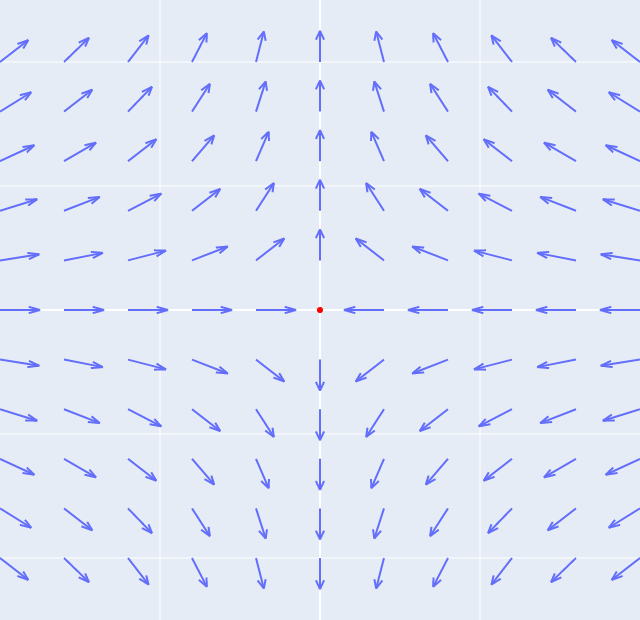}
	\end{subfloat}
	\caption{
		Left side: A centre region, the field lines wrap around in circles around a singularity of the field (red dot).
		Right side: A saddle region, around the red dot.}
	\label{1F1}
\end{figure}

This coil optimisation procedure, developed in the late 1980s and applied to several machines,
has revealed the presence of centre and saddle-point regions, as illustrated in \Cref{1F1}.
Examples of these central regions can be found in \cite{L17,FPKB24}.
The conventional approaches to these issues involve either ignoring them and manually selecting poloidally closing contours
until the coils produce a sufficiently accurate field,
or increasing the Tikhonov regularisation weight to smooth the contours, though this sacrifices field accuracy.
Another observation is that increasing the distance between the coil winding surface and the target plasma equilibrium leads to more complex behavior \cite{LB16,KLM24}.
This distance is a key factor in stellarator optimisation: while a compact design improves cost-efficiency \cite{PV24},
tight spatial constraints introduce significant engineering challenges.
These challenges have historically complicated the manufacturing of devices such as Wendelstein 7-X \cite{Bosch13} and NCSX \cite{Dudek09}.

\begin{figure}[h]
	\centering
	\subfloat[]{%
		\includegraphics[width=0.42\textwidth, keepaspectratio,trim={0 0.5cm 0 1cm},clip]{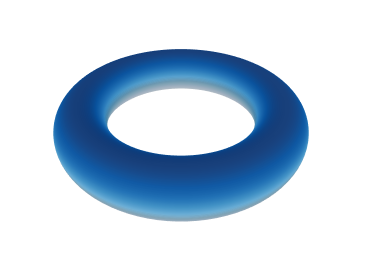}
	}
	\hspace{1cm}
	\subfloat[]{%
		\includegraphics[width=0.42\textwidth, keepaspectratio,trim={0 0.5cm 0 1cm},clip]{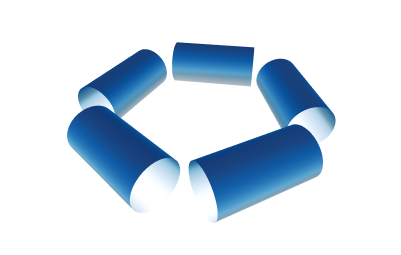}
	}
	\\
	\caption{
		Left: A toroidal coil winding surface. 
		Right side: A coil winding surface made of cylindrical structures.
	}
	\label{1F2}
\end{figure}
    
\begin{figure}[h]
	\subfloat[]{%
		\includegraphics[width=0.45\textwidth, keepaspectratio,trim={0 0.5cm 0 0.5cm},clip]{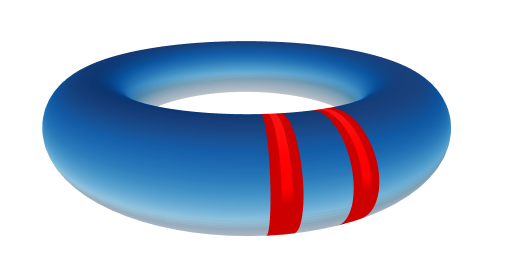}
	}
	\hspace{0.25cm}
	\subfloat[]{%
		\includegraphics[width=0.45\textwidth, keepaspectratio,trim={0 0.5cm 0 0.5cm},clip]{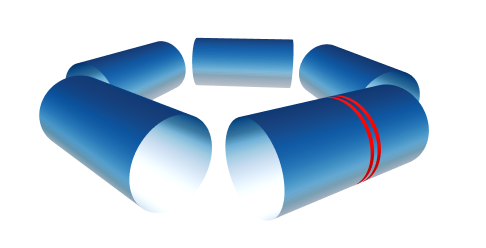}
	}
	\\
	\caption{
		Left: A toroidal coil winding surface with coils depicted by red strips.
		Right: A piecewise cylindrical coil winding surface with coils depicted by red strips.
	}
	\label{1F3}
\end{figure}

When considering the cylindrical coil approach,
the coil winding surface consists of a union of mutually disjoint cylindrical surfaces, as shown in \Cref{1F2}.
The currents are assumed to be tangent to the boundaries.
This raises the question: under what conditions do centre and saddle regions still appear?

We model our physical coils as cylindrical subsurfaces of the coil winding surface $\Sigma$, as illustrated in \Cref{1F3}.
This idealised coil model is especially accurate when using laser-engraved conductive cylinders to shape the electric current distribution needed to generate the
desired magnetic fields within the plasma confinement region,
an approach pioneered by Renaissance Fusion \cite[Magnetic chamber and modular coils patent]{RFPatents}.
Initial benchmark tests have successfully validated this proof of concept, as detailed in \cite{PPPV24}.

\subsection{Results}
\label{Sub12}

As we shall see, when the coil winding surface is assumed to be toroidal, then "generically" certain structures will be observed.
To be more precise, we obtain the following result.
\begin{thm}[Generic behaviour of toroidal currents, informal version]
	\label{1T1}
	Let $T^2\cong\Sigma\subset\mathbb{R}^3$ be a smooth toroidal surface and let $j$ be a "generic" smooth, divergence-free current on $\Sigma$. Then we have the following dichotomy:
	\begin{enumerate}
		\item Either $j(p)\neq 0$ for all $p\in \Sigma$. If this is the case, then one of the following two situations may occur:
		\begin{enumerate}
			\item All orbits of $j$ are dense within $\Sigma$.
			\item All orbits of $j$ are periodic and any two such orbits wind the same number of times in poloidal and toroidal directions before they close again.
		\end{enumerate}
		\item Or $j(p)=0$ for some $p\in \Sigma$. Then the flow of $j$ admits at least one centre- as well as saddle-point region, c.f. \Cref{1F1}. In addition, in this case all but finitely many orbits of $j$ are proper periodic orbits. The remaining orbits consist of the singular points of $j$ and non-periodic orbits which connect saddle points.
	\end{enumerate} 
\end{thm}
\begin{rem}
	\begin{enumerate}
		\item Here, "generic" means that in any given $C^1$-neighbourhood of any given smooth current $j$, we can find a current with the above mentioned dichotomy property and that all currents in some $C^1$-neighbourhood around a generic current, remain generic. More details can be found in \Cref{S3}.
		\item For a more precise description of generic field line behaviour see \Cref{6T7}.
		\item Statement (i) in \Cref{1T1} remains valid for non-generic currents, i.e. if $\Sigma\cong T^2$ and $j$ is any div-free, smooth current on $\Sigma$ which is no-where vanishing, then $j$ satisfies either statement (a) or (b) of \Cref{1T1}. We refer the reader also to \cite{Ga76A}, where statement (i) of \Cref{1T1} was proven under an additional hypothesis, as well as to \cite[Theorem 9]{PPS22} where (i) was established under a more general, so-called, winding condition.
		\item Regarding the behaviour of area preserving flows on closed surfaces of arbitrary genus, we refer the reader also to \cite{Ga76B}.
	\end{enumerate}
\end{rem}
As for cylindrical surfaces, we have the following result
\begin{thm}[Generic behaviour of cylindrical currents, informal version]
	\label{1T2}
Let $S^1\times [0,1]\cong \Sigma\subset\mathbb{R}^3$ be a smooth cylindrical surface and let $j$ be a current on $\Sigma$ such that $j$ is tangent to and non-vanishing on $\partial\Sigma$. If the field lines of $j$ on the two boundary circles of $\partial\Sigma$ point in opposite directions, c.f. \Cref{1F4}, then generically the flow of $j$ has at least one centre and one saddle-point region. In addition, all but finitely many orbits of $j$ are proper periodic orbits. The remaining orbits consist of the singular points of $j$ and non-periodic orbits which connect saddle points of $j$.
\end{thm}

\begingroup\centering
\begin{figure}[h]
	\hspace{4.5cm}\includegraphics[width=0.435\textwidth, keepaspectratio]{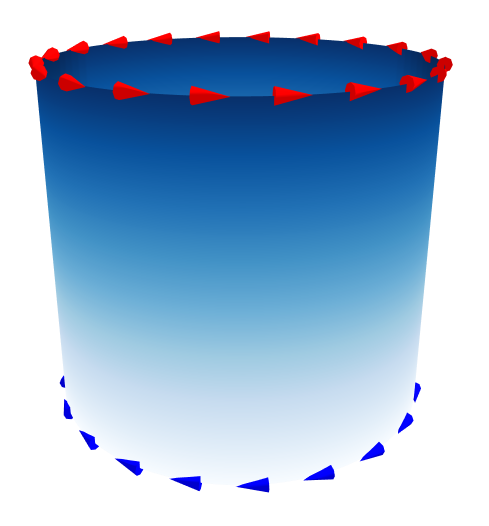}
	\caption{
		Two oppositely oriented boundary circles. Directions indicated by the arrows on a standard cylindrical surface.
		} 
	\label{1F4}
\end{figure}
\endgroup

Next, we make additional assumptions on our currents which stem from physical considerations. 
More precisely, we take into account one more of Maxwell's equations:
$\operatorname{curl}(E)=-\partial_tB$ where $E$ is the induced electric field and $B$ the magnetic field produced by the current at the coils. 
Since stellarators are intended to operate in a steady state, we assume that $\partial_tB=0$. 
Further, we assume that we are dealing with a linear isotropic medium which obeys the law $j=\sigma E$, 
where $j$ is the current producing the magnetic field and $\sigma\in \mathbb{R}$ is a material constant 
(assumed to take the same value along the whole coil winding surface). 
Combining these considerations, we find $\operatorname{curl}(j)=\sigma\operatorname{curl}(E)=-\sigma\partial_tB=0$ 
which, in turn, means that the physical currents are not only div-free on the surface but in fact also curl-free.
\newline
\newline
We note that the physical currents are then naturally assumed to be tangent to the corresponding cylindrical surfaces, since the currents are confined to the coils. Therefore, on any given cylindrical subsurface $\widetilde{\Sigma}\subset \Sigma$, the physical currents are elements of the following space
\begin{gather}
	\label{1E2}
\mathcal{H}_N(\widetilde{\Sigma}):=\left\{\gamma\in L^2(\widetilde{\Sigma},\mathbb{R}^3)\mid \mathcal{N}\cdot \gamma=0=n\cdot \gamma\text{, }\operatorname{curl}_{\widetilde{\Sigma}}(\gamma)=0=\operatorname{div}_{\widetilde{\Sigma}}(\gamma)\right\},	
\end{gather}
where $\mathcal{N}$ denotes a unit normal to $\widetilde{\Sigma}$ and $n$ denotes a unit normal to $\partial\widetilde{\Sigma}$ within $\widetilde{\Sigma}$. I.e. $\mathcal{H}_N(\widetilde{\Sigma})$ consists of the tangent fields of $\widetilde{\Sigma}$ which are div-, curl-free and remain tangent to $\partial\widetilde{\Sigma}$. This space is known as the space of harmonic Neumann fields and it is well-known to be finite dimensional and, in fact, has the same dimension as the first de Rham cohomology group of $\widetilde{\Sigma}$, c.f. \cite[Theorem 2.6.1]{S95}. We also note that elements in $\mathcal{H}_N(\widetilde{\Sigma})$ are smooth whenever $\widetilde{\Sigma}$ is smooth, \cite[Theorem 2.2.7]{S95}. Since the $\widetilde{\Sigma}$ are cylindrical surfaces, we conclude that $\dim\left(\mathcal{H}_N(\widetilde{\Sigma})\right)=1$. It is known that if we let $c_1,c_2$ denote the two boundary circles of $\partial\widetilde{\Sigma}$, then
\begin{gather}
	\label{1E3}
	\Delta \psi =0\text{ in }\widetilde{\Sigma}\text{, }\psi|_{c_1}=0\text{, }\psi|_{c_2}=1
\end{gather}
admits a unique smooth solution, \cite[Theorem 2.4.2.5]{Gris85}, and consequently it is easy to verify that
\begin{gather}
	\label{1E4}
	\gamma_0:=\mathcal{N}\times \nabla_{\widetilde{\Sigma}}\psi\in \mathcal{H}_N(\widetilde{\Sigma})\setminus \{0\}
\end{gather}
is a basis of $\mathcal{H}_N(\widetilde{\Sigma})$ where $\mathcal{N}$ is a unit normal to $\widetilde{\Sigma}$.

Further, if we let $\ell$ be some smooth curve starting at $c_1$ and ending in $c_2$, then we set, for a given current $j$, $I:=\int_{\ell}n\cdot j$ which is the electric current passing through $\ell$. Here $n$ denotes an appropriate unit normal to $\ell$ within $\widetilde{\Sigma}$ which has to be chosen such that $\{\mathcal{N},\dot{\ell},n\}$ forms a right-handed coordinate system where $\dot{\ell}$ denotes the time derivative of $\ell$, see \Cref{1F5}.

With this definition we obtain the formula
\begin{gather}
	\label{1E5}
	j=I\gamma_0\text{ for every }j\in \mathcal{H}_N(\widetilde{\Sigma}),
\end{gather}
i.e. every physical current is uniquely determined by the induced total electric current. We note that the value of $I$ is independent of the chosen curve $\ell$.

\begingroup\centering
\begin{figure}[h]
	\hspace{2cm}\includegraphics[width=.85\textwidth, keepaspectratio,trim={2cm 2cm 1cm 5cm},clip]{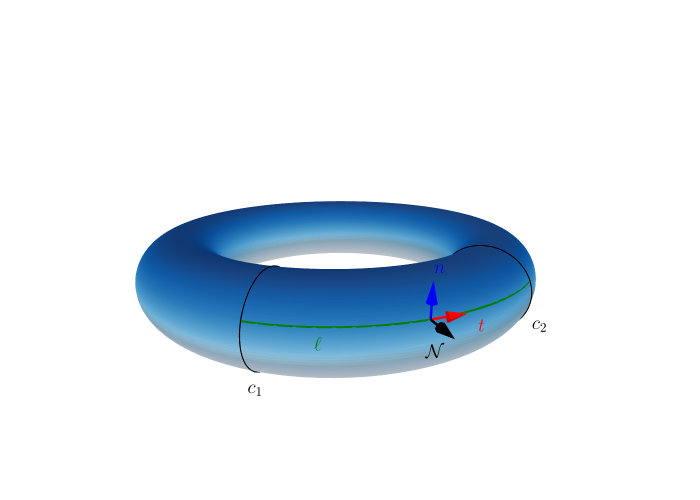}
	\caption{
		The (toroidal) coil winding surface is depicted by the blue surface.
		Two coil boundaries corresponding to $\partial\widetilde{\Sigma}$ are depicted by the black circles.
		The connecting curve $\ell$ (the circular arc) is depicted in green with a tangent vector depicted in red.
		The corresponding normal $n$ which is tangent to $\widetilde{\Sigma}$ and normal to $\ell$ corresponds to the blue upward pointing arrow.
		The remaining black arrow depicts the normal $\mathcal{N}$ which corresponds to an (outward) unit normal of the coil winding surface $\Sigma$.
		} 
	\label{1F5}
\end{figure}
\endgroup

Our main result regarding the behaviour of physical currents is as follows
\begin{thm}[Behaviour of physical currents, informal version]
	\label{1T3}
	Let $S^1\times [0,1]\cong \widetilde{\Sigma}\subset\mathbb{R}^3$ be a smooth cylindrical surface and let $j\in \mathcal{H}_N(\widetilde{\Sigma})$. Then either $j$ is identically zero, or otherwise is no-where vanishing. Further, if $j$ is not identically zero, then every field line of $j$ is a periodic poloidal orbit and in particular, the flow of $j$ does not admit any centre or saddle regions.
\end{thm}
Here poloidal means that the orbits wrap once fully around the cylinder or, in other words, that the orbits are not contractible, i.e., they represent non-trivial elements of the first fundamental group of $\widetilde{\Sigma}$.
\newline
\newline
We conclude our discussions about the dynamics of surface currents by including some further results regarding area-preserving flows and low-dimensional dynamical systems. The following is an immediate consequence of Poincar\'{e}'s famous recurrence theorem, see for instance \cite[Theorem 4]{GoHe49}, where we call a point $p\in \Sigma$ recurrent with respect to a flow $\psi_t$ if there exists sequences $(t^{\pm}_k)_k$ with $t^{\pm}_k\rightarrow \pm \infty$ and $\psi_{t^{\pm}_k}(p)\rightarrow p$.
\begin{thm}[Recurrent dynamics-Poincar\'{e} recurrence, {\cite[Theorem 4]{GoHe49}}]
	\label{1T4}
	Let $\Sigma\subset\mathbb{R}^3$ be a smooth, compact surface with possibly non-empty boundary. Let further $j\parallel \Sigma$ be a smooth tangent field with $\operatorname{div}_{\Sigma}(j)=0$ and such that $j\parallel \partial\Sigma$ in case $\partial\Sigma\neq\emptyset$. Then there exists a subset $S\subset \Sigma$ of zero area such that every $p\in \Sigma\setminus S$ is recurrent with respect to the flow of $j$.
\end{thm}
In the case of the cylinder, \Cref{1T4} can be significantly sharpened. This result will be a consequence of the Poincar\'{e}-Bendixson theorem in combination with a modification of a proof of a related result, c.f. \cite[Theorem 14.1.1]{HasKa95}, which asserts that if all points of a given orbit of a vector field (not necessarily div-free) defined on an open subset of $S^2$ are positively recurrent, then the corresponding orbit must be periodic.
\begin{thm}[Almost all orbits are periodic, informal version]
	\label{1T5}
	Let $S^1\times [0,1]\cong \Sigma\subset\mathbb{R}^3$ be a smooth cylindrical surface and $j\parallel \Sigma$ be a smooth tangent field such that $\operatorname{div}_{\Sigma}(j)=0$ and $j\parallel \partial\Sigma$. Then there exists some $S\subset \Sigma$ of zero area such that every orbit of $j$ starting at some $p\in \Sigma\setminus S$ is periodic. In addition, if $j(p)\neq 0$ for all $p\in \Sigma$, then all orbits of $j$ are poloidal periodic orbits.
\end{thm}
In the statement of \Cref{1T5} fix points are also considered as periodic orbits.
\begin{rem}
	\label{1R6}
	\begin{enumerate}
		\item We note that under the assumption $\partial\Sigma=\emptyset$ and some additional assumptions on the vector field $j$, a sharpening of \Cref{1T4} is possible, c.f. \cite[Theorem 1]{CJL99}.
		\item The result in \Cref{1T5} is optimal in the sense that one can construct div-free fields $j$ tangent to the boundary of a cylinder which support either only contractible orbits, or only poloidal orbits, or orbits of both types, as well as may admit non-periodic orbits, so that no additional information beyond what is contained in \Cref{1T5} is known a priori for a general div-free field.
		\item We emphasise that the results in \Cref{1T1} and \Cref{1T2} do not hold for all div-free tangent fields with the corresponding boundary conditions, but indeed only generically so that in this sense these results are also optimal. To see this, we can first consider a standard rotationally symmetric torus and let $\phi$ and $\theta$ denote the toroidal and poloidal angles, respectively. We can then (with a slight abuse of notation) consider any function $f(\phi,\theta)=f(\theta)$ and observe that $j:=(\partial_{\theta}f)e_{\phi}$, where $e_{\phi}$ is the normalised basis vector in toroidal direction, is div-free for any choice of such an $f$. It is clear that such a vector field cannot support a centre region because it has a vanishing poloidal component. Additionally, $\partial_{\theta}f$ must vanish somewhere since $f$ must admit a global extremum and therefore no such vector field $j$ is of type (i) nor of type (ii) as described in \Cref{1T1}. As for cylindrical surfaces, we can consider the standard rotationally symmetric cylindrical surface centered around the $z$-axis of height $2$ ($-1\leq z\leq 1$) and radius $1$. We can consider the vector field $j:=z e_{\theta}$, where $e_{\theta}$ once more denotes the (normalised) coordinate field in poloidal direction. 
		Since the boundary circles correspond to $z=-1$ and $z=1$ respectively, $j$ is non-vanishing on the boundary and of opposite orientation on the boundary circles. However, the flow of $j$ cannot admit any centre regions since $j$ points solely in poloidal direction. 
	\end{enumerate}
\end{rem}
\subsection{Comparison to previous results in the literature}
Area preserving flows have been subject of mathematical interest for a long time and we do not attempt to provide a full list of references. However, as a representative publication, we want to point out reference \cite{MaWa01}, since it gives a rather complete and general classification theorem of flows of div-free vector fields with non-degenerate zeros on general compact $2$-manifolds with or without boundary, c.f. \cite[Theorem 3.1, Theorem 4.1, Theorem 5.1]{MaWa01}. Our results complement the classification in \cite{MaWa01} by providing additional details in the specific situations relevant in plasma physics. For instance, our detailed structure theorem for the torus, c.f. \Cref{6T7}, states that if $j$ has no zeros, then its flow is semi-rectifiable (which will be a consequence of a result in \cite{PPS22}), as well as that any two orbits contained within certain regions are necessarily ambient isotopic to each other. The latter statement may be useful to characterise toroidal flows further by means of the toroidal and poloidal windings of its field lines around the toroidal surface, which are a degree of freedom in the current reconstruction process. Further, the approach we take is distinct from \cite{MaWa01} and exploits a Lojasiewicz-type inequality satisfied by the underlying vector field so that our approach is more flexible and not restricted to div-free fields (even though, of course, we exploit the div-free property to derive additional information about the flows). The structure of area-preserving flows of non-degenerate div-free fields on subsets of $S^2$ has been studied in \cite[Theorem 2.5]{MaWa02}. Again, the arguments we use differ from those in \cite{MaWa02} and, at the same time, since we restrict attention to the annulus, we can provide a more detailed characterisation. Regarding our \Cref{1T3}, we fully exploit that $j$ is harmonic, tangent to the boundary as well as the fact that $\widetilde{\Sigma}$ is an annulus which allows a complete (simple) characterisation of any such flow. Possibly because \Cref{1T3} restricts attention to such a special situation, such a precise characterisation is difficult to find in the literature. Our last \Cref{1T5} resembles a classical exercise in a dynamical systems book, but we did not find it in standard references such as \cite{HasKa95}. For this reason, we decided to include it for completeness and future reference.

Throughout, we attempted to provide detailed proofs, which are often omitted in the literature, as long as they appeared to be of reasonable length in order to make this paper more self contained and to be more accessible to people who are possibly less familiar with the theory of dynamical systems but for whom the results of the present manuscript may be of interest.
\subsection{Structure of the paper}
\label{1S3}
In section 2, we introduce some of the mathematical notation which we will use throughout this manuscript.
In section 3, we give precise definitions of the concepts which are necessary to formulate the main results and to clarify the meaning of certain notions which we will refer to when discussing numerical observations.
In section 4, we describe the numerical observations of the centre and saddle structures which motivated the theoretical results of this manuscript as well as discuss the relevance of the theoretical results for practical applications in plasma physics.
Section 5 contains the proofs of the main results \Cref{1T1},\Cref{1T2},\Cref{1T3} and \Cref{1T5}.
\Cref{A} discusses a more general version of \Cref{1T2} whose conclusion also involves a dichotomy akin to \Cref{1T1}, see \Cref{ExtraT2} for the details.

\section{Notation}
\label{S2}

We will denote by $\Sigma\subset\mathbb{R}^3$ a smooth, connected, compact surface with (possibly empty or disconnected) boundary. Usually in our applications we will have $T^2\cong \Sigma$ or $S^1\times [0,1]\cong \Sigma$ where $S^1$ denotes the unit circle and $T^2$ denotes the $2$-dimensional torus and $\Sigma$ will be embedded in $\mathbb{R}^3$. Given a smooth surface $\Sigma$, we denote by $\mathcal{V}(\Sigma):=\{X\in C^{\infty}(\Sigma,\mathbb{R}^3)\mid X\parallel \Sigma\}$ the smooth tangent fields on $\Sigma$. We denote by $\mathcal{N}$ a smoothly varying unit normal to $\Sigma$ (usually chosen to be the outward unit normal in case $\Sigma$ is a closed surface). The surface gradient of a function $\phi\in C^{\infty}(\Sigma)$ is defined as follows $\nabla_{\Sigma}\phi:=\nabla \Phi-(\mathcal{N}\cdot \nabla \Phi)\Phi$ where $\Phi$ is an arbitrary extension of $\phi$ into an open neighbourhood of $\Sigma$ and $\nabla$ denotes the standard Euclidean gradient. The value of $\nabla_{\Sigma}\phi$ is independent of the chosen extension. We denote by $\mathcal{V}_0(\Sigma):=\{X\in \mathcal{V}(\Sigma)\mid \operatorname{div}_{\Sigma}(X)=0\}$ the divergence-free tangent fields on $\Sigma$, where $X\in \mathcal{V}(\Sigma)$ is said to be divergence-free if $\int_{\Sigma}X\cdot \nabla_{\Sigma} \phi d\sigma=0$ for every $\phi\in C^{\infty}_c(\operatorname{int}(\Sigma))$, where $\operatorname{int}(\Sigma)$ denotes the manifold interior. For a given $X\in \mathcal{V}(\Sigma)$, we say that $X$ is curl-free on $\Sigma$, in symbols $\operatorname{curl}_{\Sigma}(X)=0$, if $\int_{\Sigma}X\cdot (\mathcal{N}\times \nabla_{\Sigma} \phi)d\sigma=0$ for all $\phi\in C^{\infty}_c(\operatorname{int}(\Sigma))$.
\newline
Given a smooth tangent field $X\in \mathcal{V}(\Sigma)$ with $X\parallel \partial\Sigma$ we denote by $\psi_t$ its flow, i.e., for fixed $t\in \mathbb{R}$ and any $p\in \Sigma$, $\psi_t(p)$ denotes the field line of $X$ starting at $p$ at time $t$. Sometimes we will denote for fixed $p\in \Sigma$ by $\gamma_p$ the field line of $X$ starting at $p$ so that $\gamma_p(t)=\psi_t(p)$ and the distinct notations are used to emphasise that we either consider the specific curve $\gamma_p:\mathbb{R}\rightarrow\Sigma$ or the induced diffeomorphism $\psi_t:\Sigma\rightarrow\Sigma$.
\newline
For a given vector field $X\in \mathcal{V}(\Sigma)$, $X\parallel \partial\Sigma$, we denote for given $p\in \Sigma$ the corresponding $\omega$-limit set by $\omega(p):=\{q\in \Sigma\mid \exists (t_k)_k\subset\mathbb{R}\text{ with }t_k\rightarrow\infty\text{ and }\gamma_p(t_k)\rightarrow q\}$, where $\gamma_p$ denotes the corresponding field line of $X$ starting at $p$. The $\alpha$-limit set is defined correspondingly with the only difference that we consider sequences $(t_k)_k$ diverging to $-\infty$ and is denoted by $\alpha(p)$.
\newline
If $X\in \mathcal{V}(\Sigma)$ has only isolated singularities, i.e., only isolated zeros, within $\operatorname{int}(\Sigma)$, then for any given $p\in \{q\in \operatorname{int}(\Sigma)\mid X(q)=0\}$ we denote its index by $\operatorname{ind}_p(X)$.

\section{Mathematical preliminaries}
\label{S3}
We start by giving some definitions.
\begin{defn}[Centre regions]
	\label{3D1}
	Let $\Sigma$ be a connected, compact, smooth surface with (possibly empty or disconnected) boundary and let $X\in \mathcal{V}(\Sigma)$ be a smooth tangent field with $X\parallel \partial\Sigma$ (in case $\partial\Sigma\neq\emptyset$). We say that $X$ or the flow of $X$ has a \textit{centre region} if there is some $p\in \operatorname{int}(\Sigma)$ with $X(p)=0$ and such that every open neighbourhood $U$ of $p$ contains a (possibly smaller) neighbourhood $W\subseteq U$ of $p$ such that $W\setminus \{p\}$ is a union of contractible proper periodic orbits, i.e., non-constant periodic orbits which can be shrank to a point, of $X$. Given such a $p$, we call $U\subset \Sigma$ \textit{the centre region containing $p$} if $p\in U$, $U$ is a connected open set such that $U\setminus \{p\}$ consists entirely of contractible, proper periodic orbits of $X$ and $U$ is maximal (with respect to inclusion) among all such open neighbourhoods. 
\end{defn}
The next definition concerns the saddle point region. We recall here that a zero $p\in \Sigma$ of a vector field $X\in \mathcal{V}(\Sigma)$ is called non-degenerate if in some (and hence every) fixed coordinate system the Jacobian of the local representation of $X$ at $p$ is invertible. Note that according to the inverse function theorem all non-degenerate zeros are isolated. A (real) $2\times2$ matrix is said to admit a hyperbolic saddle singularity at $0$ if it admits two real, non-zero eigenvalues of opposite sign.
\begin{defn}[Saddle region]
	\label{3D2}
	Let $\Sigma$ be a connected, compact, smooth surface with (possibly empty or disconnected) boundary and let $X\in \mathcal{V}(\Sigma)$ be a smooth tangent field with $X\parallel \partial\Sigma$. Given some $p\in \operatorname{int}(\Sigma)$ with $X(p)=0$, we say that some open neighbourhood $p\in U\subset\operatorname{int}(\Sigma)$ around $p$ is a \textit{saddle region} if $U$ is contained in the domain of a coordinate chart $\mu$ around $p$ and if there exists a $C^1$-diffeomorphism $\eta:U\rightarrow W$ onto some open subset of $\mathbb{R}^2$ with the property that $\eta(p)=0$ and $\eta_*X=A$ where $A$ is a $2\times2$-matrix with a hyperbolic saddle singularity at $0$.
\end{defn}
\begin{rem}
	\label{3R3}
	\begin{enumerate}
		\item $\eta_*X$ denotes the corresponding push-forward vector field.
		\item The condition $\eta_*X=A$ is equivalent to the condition that for given $q\in U$, if we let $x:=\eta(q)$ and denote by $\gamma^X_q$ the field line of $X$ starting at $q$ and by $\gamma^A_x$ the field line of the vector field induced by the linear operator $A$, then $\gamma^X_q=\eta^{-1}\circ \gamma^A_x$. Further, upon shrinking the domain if necessary and upon composing $\eta$ with the map $\widetilde{\eta}:\mathbb{R}^2\rightarrow\mathbb{R}^2$, $(x,y)\mapsto (x,-y)$, we may assume that $\eta$ is orientation preserving. In particular, then $\gamma^X_q$ has the same qualitative behaviour as $\gamma^A_x$. The qualitative behaviour of such a field line $\gamma^A_x$ and henceforth of $\gamma^X_q$ is depicted in \Cref{1F1}.
		\item According to the Hartman-Grobman theorem there exists a homeomorphism transforming the field lines of $X$ into field lines of its linearisation in local coordinates whenever the Jacobian of $X$ (in some chosen coordinate system) at the point $p$ has eigenvalues with non-vanishing real part. Further, in the special case of $2$-dimensional systems this homeomorphism may be taken to be of class $C^1$, c.f. \cite[Part III]{H60}.
	\end{enumerate}
\end{rem}
	After introducing the precise mathematical meanings of the dynamical structures of interest we now come to the notion of Morse functions which will allow us to talk about "generic" current distributions.
	\begin{defn}[Morse functions]
		\label{3D4}
		Let $\Sigma$ be a connected (not necessarily compact) smooth surface with empty boundary. Then a function $f:\Sigma\rightarrow\mathbb{R}$ is a \textit{Morse function} if all critical points of $f$ are non-degenerate, i.e. if for all $p\in \Sigma$ for which in some local coordinate system $\mu$ we have $\partial_1f(p)=0=\partial_2f(p)$, we necessarily have that the corresponding Hessian at $p$ is invertible.
	\end{defn}
	Morse functions will be the key element in constructing "generic" subsets of the space of div-free fields which admit nice properties. To be more precise regarding the notion of "genericity" we give the following definition.
	\begin{defn}[Genericity]
		\label{3D5}
		Let $\Sigma$ be a smooth, connected, compact surface with (possibly empty or disconnected) boundary. Suppose we are given some subset $B\subseteq \mathcal{V}(\Sigma)$. Then we say that a subset $S\subseteq B$ is \textit{($C^1$-)generic} if it is relatively open and dense in $B$ with respect to the $C^1$-topology.
		\newline
		Further, a property is said to hold \textit{($C^1$-)generically} among elements of $B$ if there exists a generic subset $A\subseteq B$ such that the property holds for all elements in $A$.
	\end{defn}
	Note that for any given subset $B\subset \mathcal{V}(\Sigma)$ it is the case that $B$ is a generic subset of itself, so that if a property holds for all elements in $B$ then it also holds generically.
	\newline
	An important example in a different context is the property of a smooth function to be a Morse function, i.e. generically every smooth function on a compact smooth manifold is a Morse function \cite[Proposition 1.2.4 \& Theorem 1.2.5]{AD14}. We will utilise the notion of Morse functions in order to obtain suitable generic sets of the div-free tangent fields on a surface which have nice properties.
	\newline
	\newline
	The next two notions will be used to provide a detailed characterisation of generic div-free flows on the torus.
	\begin{defn}[Ambient isotopy (relative to a set)]
		\label{3D6}
		Let $M$ be a smooth, connected $n$-manifold without boundary. Let $\gamma_i:S^1\rightarrow M$, $i=1,2$, be two topological embeddings, i.e. homeomorphisms onto their image, then we say that $\gamma_1$ and $\gamma_2$ are \textit{ambient isotopic} if there exists some continuous function $H:[0,1]\times M\rightarrow M$ with $H(0,x)=x$ for all $x\in M$ and $H(1,\gamma_1(p))=\gamma_2(p)$ for all $p\in S^1$ and such that $H(t,\cdot):M\rightarrow M$ is a homeomorphism for every fixed $t\in [0,1]$.
		
		Given some $A\subset M$ we say that $\gamma_1$ and $\gamma_2$ are \textit{ambient isotopic relative to} $A$ if the $\gamma_i$ are ambient isotopic and the ambient isotopy $H:[0,1]\times M\rightarrow M$ can be chosen such that for every fixed $t\in [0,1]$ we have $H(t,x)=x$ for all $x\in M\setminus A$.
		
		We say that an ambient isotopy is \textit{compactly supported} if it is an ambient isotopy relative to some compact subset $K\subset M$ and we say that $\gamma_1$ and $\gamma_2$ are \textit{compactly ambient isotopic} if they are ambient isotopic by means of a compactly supported ambient isotopy.
	\end{defn}
	\begin{rem}
		\label{3R7}
		The notion of an ambient isotopy defines an equivalence relation on the set of topological embeddings from $S^1$ into $M$, see for instance \cite[Proposition A.1.4 \& Corollary A.1.6]{G20Diss} for a detailed proof.
	\end{rem}
	In the next definition we identify $T^2\cong \mathbb{R}^2\slash \mathbb{Z}^2$ and let $\theta,\phi$ denote the standard induced coordinates on this representation of the torus.
	\begin{defn}[(Semi-)rectifiability]
		\label{3D8}
		Let $\Sigma\cong T^2=\mathbb{R}^2\slash \mathbb{Z}^2$ be a toroidal surface and $X\in \mathcal{V}(\Sigma)$. Then $X$ is called \textit{rectifiable} if there exists a diffeomorphism $\Psi:\Sigma\rightarrow T^2$ such that $\Psi_*X=a\partial_{\theta}+b\partial_{\phi}$ for suitable constants $a,b\in \mathbb{R}$ where $\Psi_*X$ denotes the corresponding pushforward vector field, in other words it is the representation of $X$ in the "global coordinate system" $(\theta,\phi)$ on $T^2$.
		
		A vector field $X\in \mathcal{V}(\Sigma)$ is called \textit{semi-rectifiable} if there exists some strictly positive function $f\in C^{\infty}(\Sigma,(0,\infty))$ such that $f\cdot X$ is rectifiable.
	\end{defn}

\section{Discussion and relevance of findings}
\label{S4}
Despite the advantages of stellarators over tokamaks, the complexity of their coil design has traditionally limited their adoption and development. In this context, characterising current distributions on coil winding surfaces is crucial for evaluating the feasibility of coil fabrication and assembly.

Traditional designs use coil winding surfaces at a fixed distance from the plasma boundary,
known as conformal surfaces. 
Simple current distributions on these surfaces can be approximated using non-planar filamentary coils,
essentially one-dimensional in current flow, distributed around the surface. 
For example, the Wendelstein 7-X stellarator uses non-planar coils that present significant manufacturing and assembly challenges due to their complex geometries, tight precision demands,
and large scale \cite{Bosch13}.

Promising advances in superconducting technology now enable the use of
broad superconducting surfaces, potentially simplifying the design, manufacturing,
and assembly of stellarator coils \cite{RFPatents}. 
While these wide surfaces allow for simplified coil winding geometries, 
such as piecewise cylindrical surfaces, they require managing more complex two-dimensional
current distributions.
On these surfaces, specfic current patterns can be achieved using laser-engraved
grooves to guide the current along channels shaped to minimise magnetic
field errors in the plasma region.
This approach has already been explored in other applications, demonstrating the feasibility
of patterning conductive surfaces to achieve precise magnetic field configurations \cite{PPPV24}.

While the patterned surface approach simplifies manufacturing,
the complexity of the resulting current patterns remains a key consideration. In particular, the presence of centre and saddle-point regions introduces engineering difficulties,
especially regarding current feeding and the potential need for dedicated access points, which tend to further degrade magnetic field accuracy. 
Managing these regions effectively requires significant design efforts to minimise their impact on overall system performance.

\subsection{Numerical observations}
\label{Sub41}

In this section, we present numerical examples of the current patterns observed on a piecewise cylindrical coil winding surface.
The results shown are based on a scaled-down Wendelstein 7-X equilibrium using five identical modules (one per field period).

These numerical examples illustrate the competition in \ref{1E1} that is central to stellarator coil design: higher regularisation simplifies the current distribution, easing manufacturing, but reduces magnetic field accuracy.
For small values of $\lambda$, the magnetic field error term $\chi_B^2$ dominates, leading to complex current patterns with many centre and saddle regions.
As $\lambda$ inscreases, the regularisation term $\chi_j^2$ becomes more important.
The toroidal circulation of the magnetic field is given by the equilibrium, and results in a fixed net poloidal current according to Ampère's law.
Minimising the current norm then results in a uniform poloidal current, which is the simplest way to ensure a net poloidal current, see the discussion starting at the bottom of page 7 of \cite{G25BS} for a more detailed discussion about how the poloidal part of $j$ is fixed.
Balancing these factors is critical for the practical implementation
of stellarator coils, especially with patterned conducting surfaces.

\Cref{complex_current} shows an optimal current distribution obtained with minimal regularisation ($\lambda=10^{-20}$ in \eqref{1E1}) on a cylindrical coil winding surface (single module).
the distribution features numerous centre and saddle regions across nearly the entire domain,
highlighting potential engineering challenges that could offset the
benefits of patterned surfaces, such as the need for multiple individual current feeds per region. 
Similar challenges arise in the traditional conformal approach,
where numerous centre and saddle regions complicate the design and assembly of filamentary coils.

\Cref{orbits_observation} shows the effect of increasing the regularisation factor to $\lambda = 10^{-15}$.
Most streamlines close in the poloidal direction, with only a few centre and saddle regions present.
This higher regularisation produces a simpler current pattern, but at the cost of a less
accurate approximation of the target magnetic field.

\begingroup
\begin{figure}[H]
\centering
	\includegraphics[width=.9\textwidth,trim={0.3cm 0 0 0},clip]{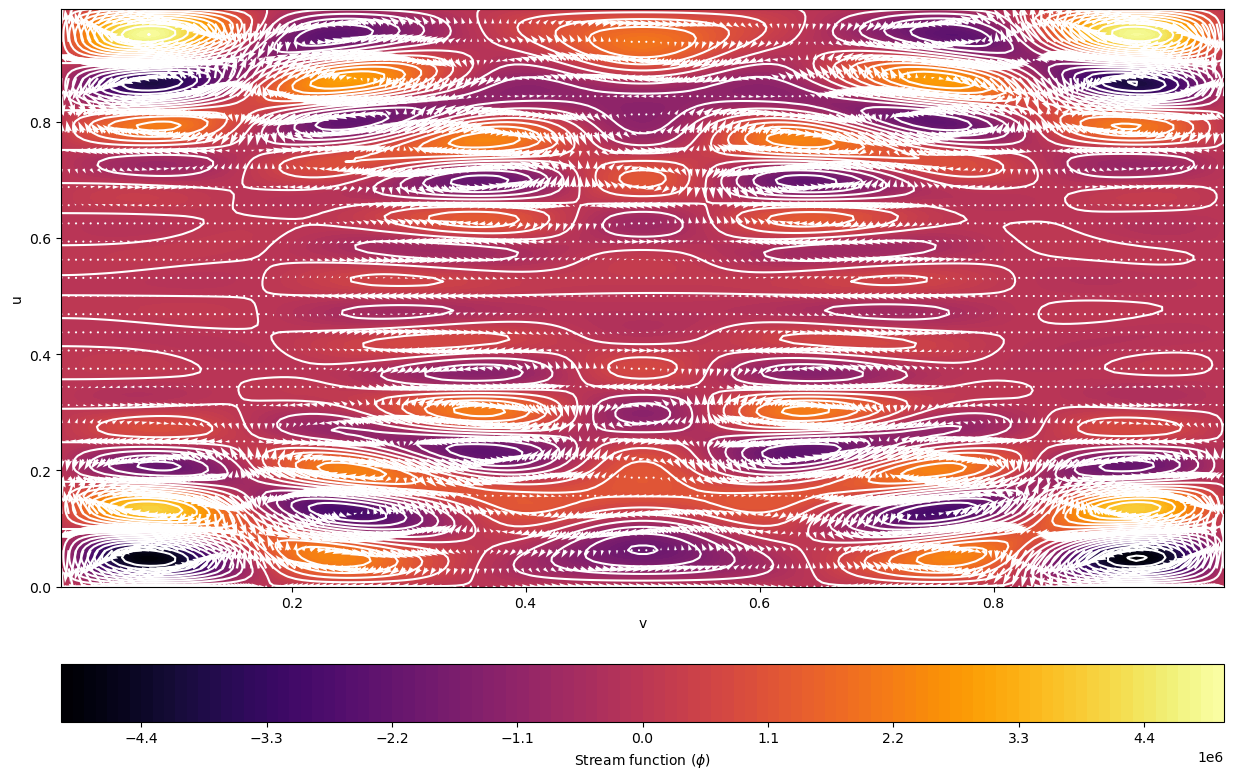}
	\caption{
		Optimal electric current distribution with minimal regularisation on a 2D parametric
		domain representing one module of a piecewise cylindrical coil winding surface.
		Poloidal direction around the plasma is $u$,
		toroidal direction along the plasma is $v$.
		The colormap shows the local current potential; white lines indicate current streamlines,
		and white arrows show the current field direction.
		Numerous centre and saddle regions are visible.} 
	\label{complex_current}
	\includegraphics[width=.9\textwidth,trim={0.3cm 0 0 0},clip]{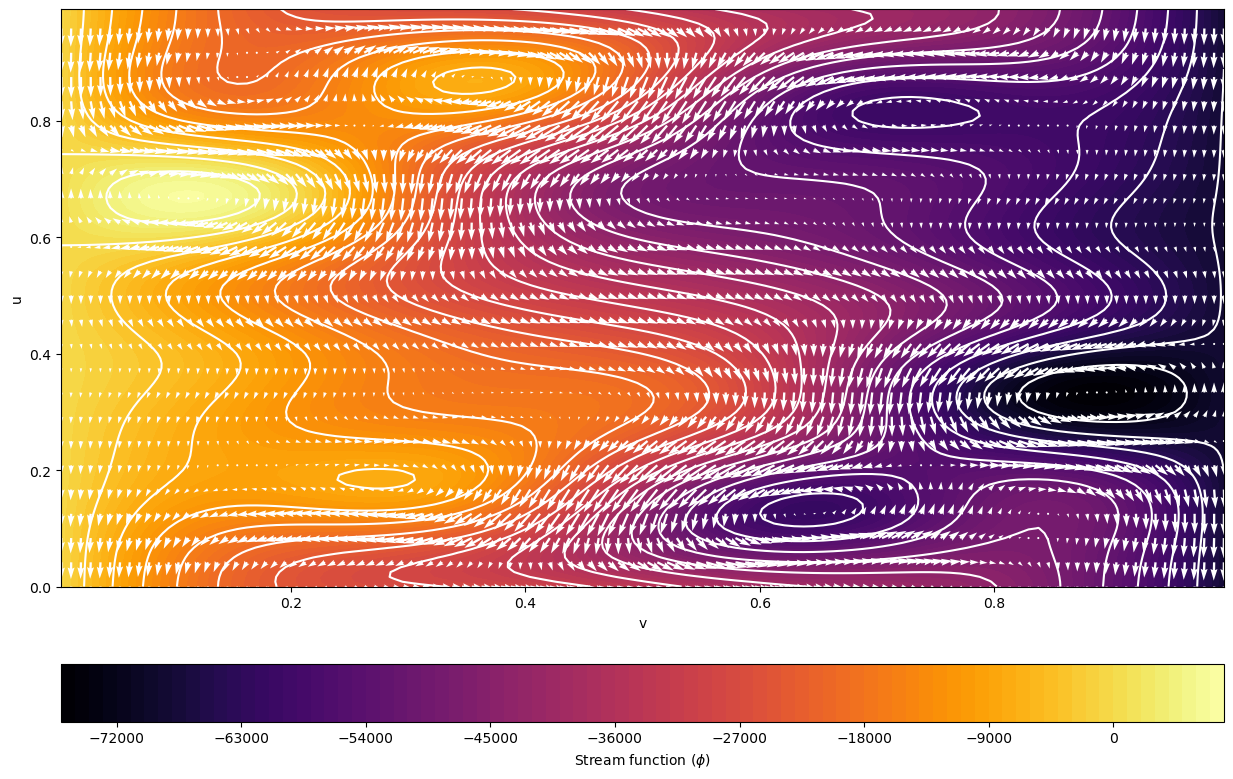}
	\caption{
		Optimal electric current distribution with low regularisation on a 2D parametric
		domain representing one module of a piecewise cylindrical coil winding surface.
		Poloidal direction around the plasma is $u$,
		toroidal direction along the plasma is $v$.
		The colormap shows the local current potential; white lines indicate current streamlines,
		and white arrows show the current field direction.
		A moderate number of saddle and contractible orbit zones are observed.} 
	\label{orbits_observation}
\end{figure}
\endgroup

\begingroup
\begin{figure}[H]
    \centering
	\includegraphics[width=.9\textwidth,trim={0.3cm 0 0 0},clip]{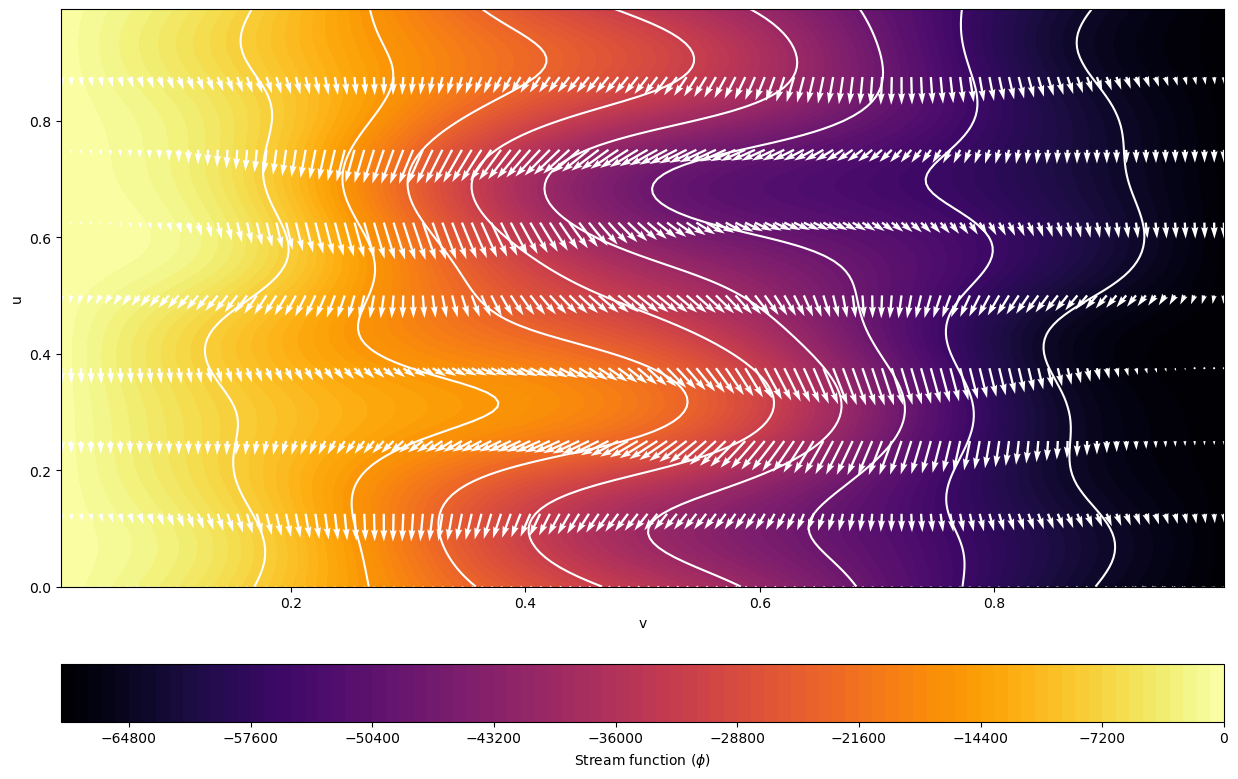}
	\caption{
		Optimal electric current distribution with high regularisation on a 2D parametric
		domain representing one module of a piecewise cylindrical coil winding surface.
		Poloidal direction around the plasma is $u$,
		toroidal direction along the plasma is $v$.
		The colormap shows the local current potential; white lines indicate current streamlines,
		and white arrows show the current field direction.
		All streamlines close poloidally.} 
	\label{poloidal_current}
	\includegraphics[width=.9\textwidth,trim={0.3cm 0 0 0},clip]{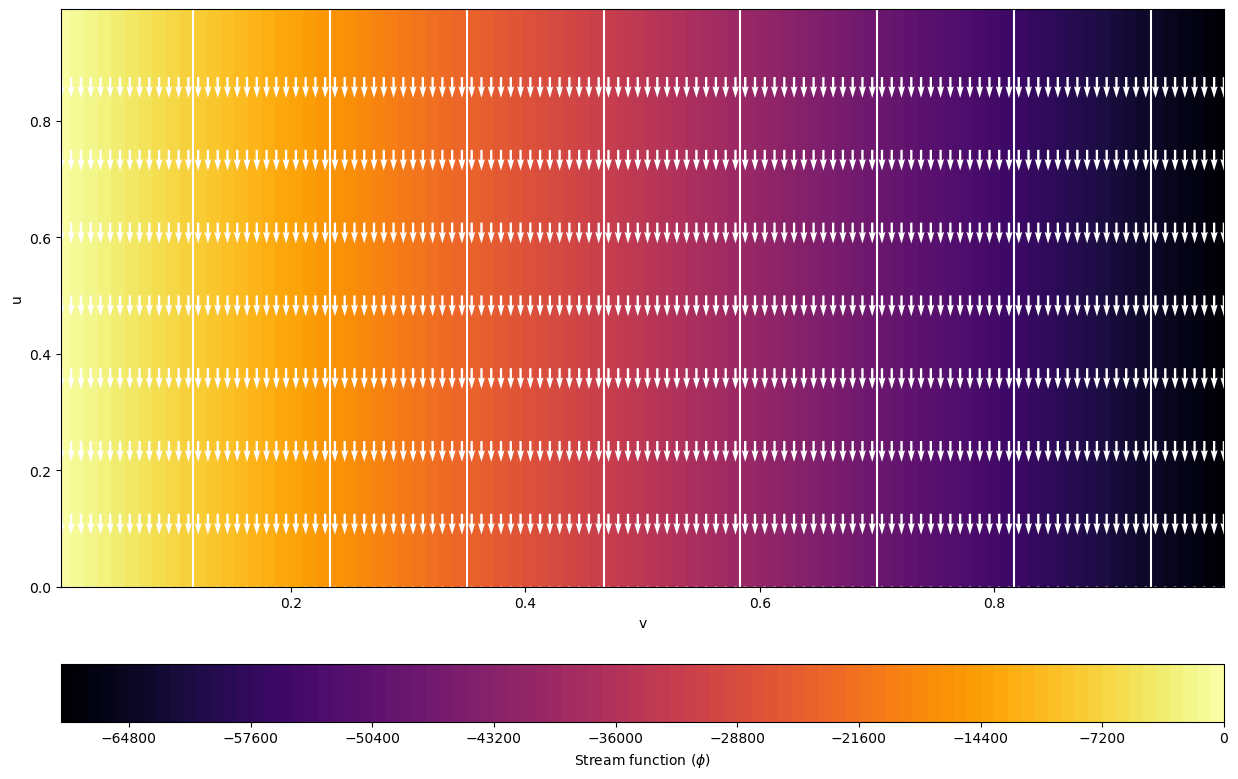}
	\caption{
		Optimal electric current distribution with a high enough regularisation on a 2D parametric
		domain representing one module of a piecewise cylindrical coil winding surface. 
		Poloidal direction around the plasma is $u$,
		toroidal direction along the plasma is $v$.
		The colormap shows the local current potential; white lines indicate current streamlines,
		and white arrows show the current field direction.
		Uniform poloidal current distribution.} 
	\label{straight_current}
\end{figure}
\endgroup

\Cref{poloidal_current} shows that above a threshold, here $\lambda=10^{-14}$,
all current streamlines close poloidally.
This solution leads to a straightforward coil design for wide patterned
conductive surfaces, but it does not provide sufficient magnetic field accuracy.

\Cref{straight_current} shows that with sufficiently high regularisation, here $\lambda=10^{-5}$,
the current constraint \eqref{chi_j} dominates the magnetic field constraint
\eqref{chi_B}, leading to a uniform poloidal current distribution around the cylinder.

\subsection{Application of the results to stellarator optimisation}
\label{Sub42}

In the coil winding surface approach, once a continuous current distribution is found that
produces a magnetic field sufficiently close to the target, the next step is to identify physically realisable coil
structures.

Traditionally, this is done by extracting poloidally closing current streamlines,
yielding a set of non-planar coils.
An alternative is to extract streamlines and make them non-conductive, so that the current flows between them.
In both cases, understanding the streamline dynamics is essential.

Although theory suggests that centre and saddle-point regions might not appear on a toroidal
surface, as discussed in \Cref{1T1}, in practice, specifically in case $ii)$, such regions do emerge.
Similarly, for piecewise cylindrical surfaces, these
centre and saddle-point regions are generally expected, as shown in \Cref{1T2}.
This result initially assumes non-vanishing and oppositely oriented currents
at the cylinder boundaries but has since been generalised to the weaker condition of
merely non-vanishing currents on these boundaries, as detailed in \Cref{A}.

Between two poloidal grooves, as anticipated, \Cref{1T3} confirms the absence of centre or saddle-point regions, with all streamlines following poloidal orbits.
To replicate the dynamics of centre regions,
an alternative to poloidal grooves must be explored.
One possible solution within these regions is to adopt a different figure of merit.
We note that the Biot-Savart operator is continuous with respect to the
$W^{-\frac{1}{2},2}(\Sigma)$-norm (and thus also with respect to the $L^2(\Sigma)$-norm),
as shown in \cite[Lemma C.1]{G24}. 
One can then focus on minimising the $W^{-\frac{1}{2},2}$-norm (or the $L^2$-norm) between the
physical and continuous currents within contractible orbit regions.
For regions composed of poloidal curves, coils can be aligned with the poloidal field lines of the continuous current distribution. 
However, proximity in $L^2(\Sigma)$ does not guarantee similar dynamics.
Thus, the presence of centre regions does not inherently prevent achieving a
good approximation with poloidal physical currents.

The results presented here are especially relevant as stellarator coil technology advances beyond
traditional filamentary designs.
With the development of patterned superconducting surfaces that can directly support surface currents,
the distinction between continuous currents and coils is becoming less distinct.

In these concepts, the coils are defined by the regions between laser-engraved grooves,
which can be connected in series, in parallel, or powered individually. 
The groove pattern should align with the streamlines of the computed optimal current distribution.
If the optimal current potential primarily produces broad poloidal currents, the resulting pattern may feature several poloidally closing grooves, relatively simple to implement. 
However, if the current potential contains many localised contractible orbits,
the surface must incorporate a network of separate etched loops, adding significant engineering complexity.

Crucially, regularisation allows one to control the density and fineness of the groove network:
a smoother current distribution produces fewer, broader conductive paths, simplifying engineering implementation.

Stronger regularisation, however, reduces the accuracy of the magnetic field approximation, a
fundamental trade-off in classical stellarator coil design.
The topology of surface currents is therefore critical in stellarator coil optimisation. 
A rigorous understanding of how centre and saddle-point regions emerge enables
the development of new approaches that effectively integrate these features when converting optimal
current patterns into discrete coils.

The practical impact of these advances is clear: by effectively managing complex continuous
current distributions, we can create coil designs that are far easier to manufacture, 
without substantially sacrificing the quality of the magnetic configuration. 
This approach lowers both costs and risks in stellarator projects while still meeting plasma performance objectives.

\section{Proofs}
\label{S6}
\subsection{Auxiliary results}
We recall that a zero $p\in \Sigma$ of a tangent field $j\parallel \Sigma$ is said to be non-degenerate if in some (and hence every) local coordinate system the Jacobian of $j$ at $p$ is non-singular.
\begin{lem}
	\label{6L1}
	Let $\Sigma\subset\mathbb{R}^3$ be a smooth, closed surface. Then the set
	\begin{gather}
		\nonumber
		\{j\in \mathcal{V}_0(\Sigma)\mid j\text{ has only non-degenerate zeros}\}
	\end{gather}
	is a generic subset of $\mathcal{V}_0(\Sigma)$.
\end{lem}
\begin{proof}[Proof of \Cref{6L1}]
	\underline{Step 1:} In a first step we prove the corresponding statement for curl-free vector fields before we transfer the results via a duality argument to div-free fields in a second step. To this end fix any $Y\in \mathcal{V}(\Sigma)$ with $\operatorname{curl}_{\Sigma}(Y)=0$. Our goal is to adapt the reasoning in the case of Morse functions, c.f. \cite[Chapter 1.2]{AD14}. Note however that not every curl-free field will be the gradient of a function, so that some additional work is necessary. We start by fixing any $\alpha>0$ and $p\in \mathbb{R}^3$. We then define the vector field
	\begin{gather}
		\label{6E1}
		F_p(x):=\alpha\nabla_{\Sigma}\frac{|x-p|^2}{2}+Y(x)\in \mathcal{V}(\Sigma).
	\end{gather}
	Further, we define the normal bundle of $\Sigma$ by $N:=\{(x,v)\mid v\perp T_x\Sigma\}\subset \Sigma\times \mathbb{R}^3$ where $T_x\Sigma$ denotes the tangent plane of $\Sigma$ at the point $x$. The normal bundle is a smooth manifold in its own right and we define the following smooth function on $N$
	\begin{gather}
		\label{6E2}
		\widehat{E}:N\rightarrow \mathbb{R}^3\text{, }(x,v)\mapsto x+v+\frac{Y(x)}{\alpha}.
	\end{gather}
	Just like in the case of Morse functions, \cite[Proposition 1.2.1]{AD14}, the goal now is to show that the points $p\in \mathbb{R}^3$ (for fixed $\alpha>0$) for which $F_p$ has degenerate zeros is contained in the set of critical values of $\widehat{E}$ so that Sard's theorem, c.f. \cite[Theorem 6.10]{L12}, will imply that this set is a null set.
	\newline
	We start by observing that
	\begin{gather}
		\label{6E3}
		F_p(x)=0\Leftrightarrow \nabla_{\Sigma}\frac{|x-p|^2}{2}=-\frac{Y(x)}{\alpha}\Leftrightarrow \alpha (x-p)^\parallel=-Y(x)
	\end{gather}
	where $\cdot^\parallel$ denotes the part of $x-p$ which is tangent to $T_x\Sigma$. Now we can fix an arbitrary coordinate system $(u_1,u_2)\mapsto x(u_1,u_2)$ of $\Sigma$ and letting $g_{ij}$ denote the metric tensor, $g^{ij}$ the corresponding inverse matrix and $Y^j$ the local expression of $Y$, we find $F_p(x)=\left(\alpha g^{ij}(x-p)\cdot \partial_ix+Y^j\right)\partial_jx$ and consequently
	\begin{gather}
		\nonumber
		\partial_kF^j_p(x)=\partial_kY^j+\alpha g^{ij}\partial_kx\cdot \partial_ix+\alpha g^{ij}(x-p)\cdot \partial_k\partial_ix+\alpha \partial_ix\cdot (x-p)\partial_kg^{ij}
	\end{gather}
	where we use Einstein's summation convention. We observe that $\partial_ix\parallel \Sigma$ and thus $\alpha\partial_ix\cdot (x-p)=\alpha \partial_ix\cdot (x-p)^\parallel=-Y(x)\cdot \partial_ix$ where we utilised (\ref{6E3}). We conclude that if $x\in \Sigma$ is a degenerate zero of $F_p$, then the following matrix is singular
	\begin{gather}
		\label{6E4}
		\partial_kY^j+\alpha \delta^j_{k}+\alpha g^{ij}(x-p)\cdot \partial_k\partial_ix-Y(x)\cdot \partial_ix \partial_kg^{ij}\text{, }1\leq j,k\leq 2
	\end{gather} 
	where we used that $\partial_kx\cdot \partial_ix=g_{ki}$ and $g^{ij}g_{ki}=\delta^j_{k}$ is the Kronecker delta. We can now on the other hand verify that $q\in \mathbb{R}^3$ is a critical value of $\widehat{E}$ if and only if $q=\widehat{E}(x,v)$ for some $(x,v)\in N$ and the following matrix is singular, the calculations are identical as in \cite[Lemma 1.2.2]{AD14} where the map $E(x,v):=x+v$ was considered,
	\begin{gather}
		\label{6E5}
		\partial_kx\cdot \partial_ix-v\cdot \partial_k\partial_ix+\frac{(\partial_k(Y^m\partial_mx))\cdot \partial_ix}{\alpha}.
	\end{gather}
	We now observe that if $x$ is a zero of $F_p$ then $\alpha(x-p)+Y(x)\perp T_x\Sigma$ and thus we may pick $v=-((x-p)+\frac{Y(x)}{\alpha})$ and note that $\widehat{E}(x,v)=p$. By our previous reasoning $p$ will be a critical value of $\widehat{E}$ if the matrix
	\begin{gather}
		\nonumber
		\partial_kx\cdot \partial_ix+(x-p)\cdot \partial_k\partial_ix+\frac{Y(x)\cdot \partial_k\partial_ix}{\alpha}+\frac{(\partial_k(Y^m\partial_mx))\cdot \partial_ix}{\alpha}
		\\
		\nonumber
		=g_{ki}+(x-p)\cdot \partial_k\partial_ix+\frac{\partial_k(Y^mg_{mi})}{\alpha}
	\end{gather}
	is singular, where we expressed $Y=Y^m\partial_mx$ and used again that $\partial_kx\cdot \partial_ix=g_{ki}$. We can multiply the above expression by $g^{ij}$ and take the sum over $i$ which corresponds to a matrix multiplication by a non-singular matrix so that the above matrix is singular if and only if
	\begin{gather}
		\label{6E6}
		\delta^j_k+g^{ij}(x-p)\cdot \partial_k\partial_ix+\frac{g^{ij}\partial_k(Y^mg_{mi})}{\alpha}
	\end{gather}
	is singular. We finally note that
	\begin{gather}
		\nonumber
		g^{ij}\partial_k(Y^mg_{mi})=g^{ij}g_{mi}\partial_kY^m+Y^mg^{ij}\partial_kg_{mi}=\partial_kY^j-Y^mg_{mi}\partial_kg^{ij}=\partial_kY^j-Y\cdot \partial_ix \partial_kg^{ij}
	\end{gather}
	where we used that $0=\partial_k\delta^{j}_m=\partial_k(g^{ji}g_{im})$ and $Y^mg_{mi}=Y^m\partial_mx\cdot \partial_ix=Y\cdot \partial_ix$. We can insert this into (\ref{6E6}) and observe that the expression obtained in that way coincides (modulo the non-zero factor $\alpha$) with the expression in (\ref{6E4}).
	\newline
	To summarise, we showed that if $x\in \Sigma$ is a degenerate zero of $F_p$, then (\ref{6E3}) holds and (\ref{6E4}) defines a singular matrix. Further, $p=\widehat{E}(x,v)$ with $v=-(x-p)-\frac{Y(x)}{\alpha}$, $(x,v)\in N$ turns out to be a critical value of $\widehat{E}$ by means of the characterisation in (\ref{6E5}). Consequently the set of all $p\in \mathbb{R}^3$ such that $F_p$ has degenerate zeros is a null set.
	\newline
	We can now consider the sequence $\alpha_n:=\frac{1}{n}$ and we observe that for every $n\in \mathbb{N}$ there exists a null set $A_n\subset \mathbb{R}^3$ such that for all $p\in \mathbb{R}^3\setminus A_n$ the vector field $F_n(x):=\frac{\nabla_{\Sigma}|x-p|^2}{2n}+Y(x)$ has only non-degenerate zeros. We can then let $A:=\cup_{n\in \mathbb{N}}A_n$ which remains a null set so that we can fix any $p\in \mathbb{R}^3\setminus A$ and find that for every $n\in \mathbb{N}$ the vector field $F_n(x)=\frac{\nabla_{\Sigma}|x-p|^2}{2n}+Y(x)$ has only non-degenerate zeros. Lastly, if $\operatorname{curl}_{\Sigma}(Y)=0$ it is clear that also $\operatorname{curl}_{\Sigma}(F_n)=0$ for all $n$ and we have
	\begin{gather}
		\nonumber
		\|F_n-Y\|_{C^1(\Sigma)}=\frac{1}{n}\left\|\nabla_{\Sigma}\frac{|x-p|^2}{2}\right\|_{C^1(\Sigma)}\rightarrow 0
	\end{gather}
	since $p$ is independent of $n$. The fact that the set of curl-free fields with non-degenerate zeros is open in the space of curl-free fields with respect to the $C^1$-topology is straightforward to confirm.
	\newline
	\newline
	\underline{Step 2:} We fix any (smooth) unit normal $\mathcal{N}$ on $\Sigma$ and we note that taking the cross product with $\mathcal{N}$ induces an isomorphism between the smooth div-free fields on $\Sigma$ and the curl-free fields on $\Sigma$ (which is most easily seen by observing that in the language of differential forms this corresponds to applying the Hodge star operator to the corresponding $1$-form). Hence, given any $j\in \mathcal{V}_0(\Sigma)$ we can write $j=\mathcal{N}\times Y$ for a suitable $Y\in \mathcal{V}(\Sigma)$ with $\operatorname{curl}_{\Sigma}(Y)=0$. By step 1 we may approximate $Y$ by curl-free tangent fields $Y_n$ in $C^1(\Sigma)$-topology which all have only non-degenerate zeros. We define accordingly $j_n:=\mathcal{N}\times Y_n\in \mathcal{V}_0(\Sigma)$ and one readily checks that $j_n\rightarrow j$ in $C^1(\Sigma)$ and that further the zero sets of the $j_n$ and the $Y_n$ coincide and that $j_n$ has a degenerate zero $x$ if and only if $x$ is a degenerate zero of $Y_n$. Consequently the $j_n$ are all div-free with only non-degenerate zeros and approximating $j$ in $C^1$-topology.
\end{proof}
The following is the corresponding result for vector fields of interest as described in \Cref{1T2}.
\begin{lem}
	\label{6L2}
	Let $S^1\times [0,1]\cong \Sigma\subset\mathbb{R}^3$ be a smooth cylindrical surface and consider the set 
	\begin{gather}
		\nonumber
		\mathcal{V}_0^{\operatorname{OT}}(\Sigma):=\{j\in \mathcal{V}_0(\Sigma)\mid j\parallel \partial\Sigma\text{, }j\neq 0\text{ on }\partial\Sigma\text{ and the field lines of }j\text{ on the two distinct} 
		\\
		\nonumber
		\text{boundary components are oppositely oriented}\}\subset \mathcal{V}(\Sigma)
	\end{gather}
	of div-free, "oppositely boundary tangent" fields. Then the following subset
	\begin{gather}
		\nonumber
		\{j\in \mathcal{V}_0^{\operatorname{OT}}(\Sigma)\mid j\text{ has only non-degenerate zeros} \}
	\end{gather}
	is a generic subset of $\mathcal{V}_0^{\operatorname{OT}}(\Sigma)$.
\end{lem}
\begin{proof}[Proof of \Cref{6L2}]
	We start by fixing some $j\in \mathcal{V}_0^{\operatorname{OT}}(\Sigma)$ and we observe that according to the Hodge-Morrey-decomposition theorem \cite[Theorem 2.4.2]{S95} we may write $j=\mathcal{N}\times \nabla \phi+\gamma$ for a suitable $\phi\in C^{\infty}(\Sigma)$ with $\phi|_{\partial\Sigma}=0$ and $\gamma\in \mathcal{H}_N(\Sigma)$ where we used that $j$ is div-free and tangent to the boundary and where $\mathcal{N}$ denotes a (smooth) unit normal to $\Sigma$. Then, according to (\ref{1E2})-(\ref{1E5}) we can therefore write $j=\mathcal{N}\times \nabla \chi$ for a suitable $\chi\in C^{\infty}(\Sigma)$ with $\chi|_{\partial\Sigma}$ is locally constant, i.e. is constant on each boundary circle of $\partial\Sigma$ with possibly distinct constants on each circle. We can then consider the outward flow along an extension of the outward normal $n$ along $\partial\Sigma$, i.e. $n\perp \partial\Sigma$, $n\cdot \mathcal{N}=0$, so that we may assume without loss of generality that $\Sigma$ is contained in a slightly larger cylindrical surface $\widetilde{\Sigma}$ and we may also extend $\chi$ to a smooth function $\widetilde{\chi}$ on $\widetilde{\Sigma}$. It then follows from standard Morse theory, \cite[Proposition 1.2.4]{AD14}, that we can approximate $\widetilde{\chi}$ on $\operatorname{int}(\widetilde{\Sigma})$ by Morse functions $f_n$ in $C^2_{\operatorname{loc}}(\operatorname{int}(\widetilde{\Sigma}))$. Consequently, since $\Sigma\subset \operatorname{int}(\widetilde{\Sigma})$ is compact, the restrictions $f_n|_{\Sigma}$ approximate $\chi$ in $C^2(\Sigma)$. We now make use of the assumption that $j\neq 0$ on $\partial\Sigma$ which is equivalent to the statement that $\nabla_{\Sigma}\chi \neq 0$ on $\partial\Sigma$. We observe that then for a suitable $0<\epsilon<\frac{1}{2}$ it will be the case that $|\nabla_{\Sigma}\chi|\geq c_0>0$ on $[0,\epsilon]\times S^1\cup [1-\epsilon,1]\times S^1$ (where we implicitly use the identification $[0,1]\times S^1\cong \Sigma$). Fix such an $\epsilon$ and fix a bump function $\rho\in C^{\infty}(\Sigma)$ satisfying $\rho|_{\partial\Sigma}=1$ and $\rho|_{[\epsilon,1-\epsilon]\times S^1}=0$ and define 
	\begin{gather}
		\nonumber
		\chi_n:=\rho \chi+(1-\rho)f_n|_{\Sigma}.
	\end{gather}
	Then
	\begin{gather}
		\nonumber
		\chi-\chi_n=(1-\rho)(\chi-f_n|_{\Sigma})\Rightarrow \|\chi-\chi_n\|_{C^2(\Sigma)}\leq C(\rho)\|\chi-f_n|_{\Sigma}\|_{C^2(\Sigma)}\rightarrow 0\text{ as }n\rightarrow\infty.
	\end{gather}
	In particular, $|\nabla_{\Sigma}\chi_n|\geq \frac{|\nabla_{\Sigma}\chi|}{2}\geq \frac{c_0}{2}>0$ on $[0,\epsilon]\times S^1\cup [1-\epsilon,1]\times S^1$ for all large enough $n$. Consequently, $\nabla \chi_n\neq 0$ on $[0,\epsilon]\times S^1\cup [1-\epsilon,1]\times S^1$. In addition, $\rho=0$ in $[\epsilon,1-\epsilon]\times S^1$ so that $\chi_n=f_n$ within this set and hence, since the $f_n$ are Morse, $\nabla \chi_n$ has only non-degenerate zeros throughout $\Sigma$. Lastly, $\chi_n|_{\partial\Sigma}=\chi|_{\partial\Sigma}$ is locally constant since $\rho|_{\partial\Sigma}=1$ and by properties of $\chi$. This most importantly implies that $\nabla_{\Sigma}\chi_n$ is normal to $\partial\Sigma$ and approximates $\nabla_{\Sigma}\chi$ in $C^1(\Sigma)$-norm. It is then straightforward to verify, for instance working in normal coordinates, c.f. \cite[Proposition 5.24]{L18}, to reduce to the standard Euclidean situation, that the zero sets of $\nabla_{\Sigma}\chi_n$ and $j_n:=\mathcal{N}\times \nabla_{\Sigma}\chi_n$ coincide and that a zero of $\nabla_{\Sigma}\chi_n$ is non-degenerate if and only if the same zero is non-degenerate for $j_n$. In addition, one easily checks that $j_n$ converges in $C^1(\Sigma)$-norm to $\mathcal{N}\times \nabla_{\Sigma}\chi=j$ (in terms of differential forms $\mathcal{N}\times \nabla_{\Sigma}\chi$ corresponds to $\star d\chi$ where $\star$ denotes the Hodge star operator and $d$ denotes the exterior derivative). Overall we found a sequence of vector fields $(j_n)_n\subset \mathcal{V}_0(\Sigma)$ with $j_n\neq 0$ on $\partial\Sigma$ and with only non-degenerate zeros inside $\Sigma$. We further recall that $\nabla \chi_n\perp \partial\Sigma$ which in turn implies $j_n=\mathcal{N}\times \nabla_{\Sigma}\chi_n\parallel \partial\Sigma$. Finally, using the identification $\Sigma\cong [0,1]\times S^1$ it is clear that since the field lines of $j$ are oppositely oriented on the boundary circles, then so must be those of the $j_n$ for all large enough $n$. This proves the density of $\left\{j\in \mathcal{V}_0^{\operatorname{OT}}(\Sigma)\mid \text{ j has only non-degenerate zeros}\right\}$ within $\mathcal{V}_0^{\operatorname{OT}}(\Sigma)$. The fact that this set is relatively open in $C^1$-topology is straightforward to verify, c.f. \cite[Theorem 1.2.5]{AD14} for the corresponding statement about Morse functions, so that we conclude the $C^1$-genericity of this set.
\end{proof}
During the course of the proof of \Cref{1T3} we will also need the following observation, where we recall that $\mathcal{H}_N(\Sigma)=\{\gamma\in \mathcal{V}(\Sigma)\mid \operatorname{div}_{\Sigma}(\gamma)=0=\operatorname{curl}_{\Sigma}(\gamma)\text{, }\gamma \parallel \partial\Sigma\}$.
\begin{lem}
	\label{6L3}
	Let $[0,1]\times S^1\cong \Sigma\subset \mathbb{R}^3$ be a smooth cylindrical surface. Let $j\in \mathcal{H}_N(\Sigma)$,  then the flow of $j$ does not contain any non-constant contractible periodic orbits within $\operatorname{int}(\Sigma)$.
\end{lem}
\begin{proof}[Proof of \Cref{6L3}]
	Assume for a contradiction that the flow of $j$ admits a contractible, non-constant periodic orbit within $\operatorname{int}(\Sigma)$. Then we can find a disc $D\subset \operatorname{int}(\Sigma)$ bounded by this orbit, i.e. $j\parallel \partial D$. We then know from (\ref{1E2})-(\ref{1E5}) that we can write $j=\mathcal{N}\times \nabla_{\Sigma}\psi$ for a suitable harmonic function $\psi$ on $\Sigma$, where $\mathcal{N}$ is a smooth unit normal on $\Sigma$. We note that $j\parallel \partial D$ implies $\nabla_{\Sigma}\psi\perp \partial D$ which in turn implies that $\psi|_{\partial D}$ is constant. Now according to the maximum principle \cite[Chapter 6.4.1 Theorem 1]{Evans10} $\psi$ will achieve its minimum and maximum on $\partial D$ and therefore $\psi$ is constant throughout all of $D$ and hence $\nabla_{\Sigma}\psi=0$ in $D$ so that in particular $j=0$ in $D$. It then follows from the unique continuation principle, c.f. \cite{AKS62}, that $j=0$ on all of $\Sigma$ contradicting the fact that the flow of $j$ admits a proper periodic field line.
\end{proof}
We will also need the following result, which we prove for completeness. We recall that for a given vector field $X\in \mathcal{V}(\Sigma)$ we can define the multiplicity $k\in \mathbb{N}_0$ of a zero $p\in \Sigma$ of $X$ by the property that in some (and hence every) given coordinate chart $\mu$ around $p$ we find $\partial^{\alpha}X^j(p)=0$ for all multi-indices $|\alpha|\leq k$ and all $1\leq j\leq 2$ and such that there is some multi-index $\tilde{\alpha}\in \mathbb{N}_0^2$ with $|\widetilde{\alpha}|=k+1$ and $\partial^{\widetilde{\alpha}}X^j(p)\neq 0$ for some $j\in \{1,2\}$. We denote the multiplicity of a zero $p$ of a vector field by $m(p)$.
\begin{lem}
	\label{6L4}
	Let $(\Sigma,g)$ be an oriented, compact, smooth Riemannian $2$-manifold with (possibly empty and disconnected) boundary. Let further $X\in \mathcal{V}(\Sigma)$ be a harmonic vector field, i.e. $\operatorname{div}_{\Sigma}(X)=0=\operatorname{curl}_{\Sigma}(X)$. Then either $X$ is identically zero or $\operatorname{int}(\Sigma)$ contains only isolated zeros of $X$ and for every $p\in \{q\in \operatorname{int}(\Sigma)\mid X(q)=0\}$ we have $\operatorname{ind}_p(X)=-(m(p)+1)$ where $\operatorname{ind}_p(X)$ denotes the index of the isolated zero $p$ of $X$.
\end{lem}
\begin{proof}[Proof of \Cref{6L4}]
	It follows first, since $X$ is harmonic, from \cite{AKS62} that either $X$ is identically zero or otherwise all of its zeros must be of finite order. Then, using again that $X$ is harmonic, it follows from \cite{B99} that $\{q\in \operatorname{int}(\Sigma)\mid X(q)=0\}$ is a countable and locally finite subset of $\operatorname{int}(\Sigma)$ consisting of isolated points.
	\newline
	Now fix any $p\in \{q\in \operatorname{int}(\Sigma)\mid X(q)=0\}$ and a normal coordinate system centred around $p$. We may then compute the index $\operatorname{ind}_p(X)$ within this coordinate system so that we may assume that $X$ is defined on an Euclidean ball $B_r(0)$. The vector field $X$ admits the coordinate expression $X=f \partial_1+h \partial_2$ and we may then compute the index of $X$ at $p$ according to the formula
	\begin{gather}
		\nonumber
		\operatorname{ind}_p(X)=\frac{1}{2\pi}\int_{\partial B_{\epsilon}(0)}\frac{f dh-h df}{f^2+h^2}
	\end{gather}
	where $\partial B_{\epsilon}(0)$ is oriented in the counter-clockwise way and $0<\epsilon\ll 1$ can be chosen arbitrarily. In particular, we may use the parametrisation $\gamma_{\epsilon}(t):=\epsilon (\cos(t),\sin(t))$ and for notational simplicity we set $\gamma(t):=\gamma_1(t)$. Since the computation of the index is independent of $\epsilon>0$, c.f. \cite[Lemma 7.1]{Ful95}, we find
	\begin{gather}
		\nonumber
		\operatorname{ind}_p(X)=\frac{\lim_{\epsilon\searrow0}}{2\pi}\int_{\partial B_{\epsilon}(0)}\frac{f dh-h df}{f^2+h^2}
		\\
		\nonumber
		=\frac{\lim_{\epsilon\searrow0}}{2\pi}\int_0^{2\pi}\epsilon\frac{f(\gamma_{\epsilon}(t))\dot{\gamma}^i(t)\partial_ih(\gamma_{\epsilon}(t))-h(\gamma_{\epsilon}(t))\dot{\gamma}^i(t)\partial_if(\gamma_{\epsilon}(t))}{f^2(\gamma_{\epsilon}(t))+h^2(\gamma_{\epsilon}(t))}dt
		\\
		\nonumber
		=\frac{\lim_{\epsilon\searrow 0}}{2\pi}\int_0^{2\pi}\frac{\left(\frac{f(\gamma_{\epsilon}(t))}{\epsilon}\right)\dot{\gamma}^i(t)\partial_ih(\gamma_{\epsilon}(t))-\left(\frac{h(\gamma_{\epsilon}(t))}{\epsilon}\right)\dot{\gamma}^i(t)\partial_if(\gamma_{\epsilon}(t))}{\left(\frac{f(\gamma_{\epsilon}(t))}{\epsilon}\right)^2+\left(\frac{h(\gamma_{\epsilon}(t))}{\epsilon}\right)^2}dt
	\end{gather}
	where $\gamma^1(t)=\cos(t)$ and $\gamma^2(t)=\sin(t)$ and we use throughout the Einstein summation convention.
	\newline
We set $n:=m(p)$ and we Taylor expand the function $\epsilon\mapsto f(\gamma_{\epsilon}(t))$ for fixed $t\in \mathbb{R}$ which yields
\begin{gather}
	\nonumber
	f(\gamma_{\epsilon}(t))=\sum_{k=0}^{n+1}\frac{\frac{d^{k}}{d\epsilon^k}|_{\epsilon=0}f(\gamma_{\epsilon}(t))}{k!}\epsilon^k+\rho(\epsilon,t)\epsilon^{n+2}
\end{gather}
where $\rho\equiv \rho(\epsilon,t)$ can be uniformly bounded in $\epsilon$ and $t$, i.e. there exists some $C>0$ such that $|\rho|\leq C$ for all $t$ and $\epsilon$. We observe that $\frac{d}{d\epsilon}f(\gamma_{\epsilon}(t))=\gamma^i(t)\partial_if(\gamma_{\epsilon}(t))$ and so $\frac{d^2}{d\epsilon^2}f(\gamma_{\epsilon}(t))=\gamma^i(t)\gamma^j(t)\partial_i\partial_jf(\gamma_{\epsilon}(t))$. Consequently we find
\begin{gather}
	\nonumber
	f(\gamma_{\epsilon}(t))=\frac{\gamma^{i_1}\gamma^{i_2}\dots \gamma^{i_{n+1}}\partial_{i_1}\partial_{i_2}\dots\partial_{i_{n+1}}f(0)}{(n+1)!}\epsilon^{n+1}+\rho\epsilon^{n+2}
\end{gather}
where we used that we are dealing with a zero of order $n$ and thus all derivatives of $f$ (and $h$) at $0$ vanish up to order $n$ and we employ Einstein's summation convention. Similarly we can expand $h(\gamma_{\epsilon})$ up to order $n+2$ and $\partial_ih(\gamma_{\epsilon}(t))$ and $\partial_if(\gamma_{\epsilon}(t))$ up to order $n+1$ and arrive at
\begin{gather}
	\nonumber
	\frac{2\pi}{n+1}\operatorname{ind}_p(X)=
	\\
	\nonumber
	\scalemath{0.89}{ \int_0^{2\pi}\frac{\gamma^{i_1}\text{..} \gamma^{i_{n+1}}\partial_{i_1}\text{..} \partial_{i_{n+1}}f(0) \gamma^{j_1}\text{..}\gamma^{j_n}\dot{\gamma}^m\partial_{j_1}\text{..}\partial_{j_n}\partial_mh(0)-\gamma^{i_1}\text{..} \gamma^{i_{n+1}}\partial_{i_1}\text{..} \partial_{i_{n+1}}h(0) \gamma^{j_1}\text{..}\gamma^{j_n}\dot{\gamma}^m\partial_{j_1}\text{..}\partial_{j_n}\partial_mf(0)}{\left(\gamma^{i_1}\text{..}\gamma^{i_{n+1}}\partial_{i_1}\text{..}\partial_{i_{n+1}}f(0)\right)^2+\left(\gamma^{i_1}\text{..}\gamma^{i_{n+1}}\partial_{i_1}\text{..}\partial_{i_{n+1}}h(0)\right)^2}dt}
	\\
	\label{6E7}
	+o(1)\text{ as }\epsilon\searrow0.
\end{gather}
We claim that the integrand appearing in (\ref{6E7}) equals $-1$ which will complete the proof.
\newline
To see this we fix any multi-index $\alpha\in \mathbb{N}_0^2$ with $|\alpha|=n$ and consider
\begin{gather}
	\label{6E8}
	\dot{\gamma}^m\partial^\alpha\partial_mf(0)=-\sin(t)\partial^\alpha\partial_1f(0)+\cos(t)\partial^\alpha\partial_2f(0)
\end{gather}
where we used the explicit expression $\gamma(t)=(\cos(t),\sin(t))$ and so $\dot{\gamma}(t)=(-\sin(t),\cos(t))$. Since $\operatorname{div}_g(X)=0$ we conclude that
\begin{gather}
	\label{6E9}
	\partial_1(\sqrt{\det{g}}f)+\partial_2(\sqrt{\det{g}}h)=0.
\end{gather}
We can now apply $\partial^{\alpha}$ on both sides of (\ref{6E9}) and observe that only the terms where $\partial^{\alpha}\partial_1$ acts entirely on $f$ and $\partial^{\alpha}\partial_2$ acts entirely on $h$ are non-zero because $p$ is a zero of order $n$. Hence the Leibniz rule implies 
\begin{gather}
	\label{6E10}
	\partial^{\alpha}\partial_1f(0)=-\partial^{\alpha}\partial_2h(0).
\end{gather}
Similarly, we recall that $\operatorname{curl}_g(X)=0$ which implies that
\begin{gather}
	\nonumber
	\partial_1(g_{12}f+g_{22}h)=\partial_2(g_{11}f+g_{12}h).
\end{gather}
Again we apply $\partial^{\alpha}$ on both sides, where once more all terms vanish except for those in which all derivatives act upon $f$ or $h$. Further, we note that by choice of our coordinate system we have $g_{ij}(0)=\delta_{ij}$ so that we arrive at the identity
\begin{gather}
	\label{6E11}
	\partial^{\alpha}\partial_1h(0)=\partial^{\alpha}\partial_2f(0).
\end{gather}
We can now insert (\ref{6E11}) and (\ref{6E10}) into (\ref{6E8}) to deduce that
\begin{gather}
	\label{6E12}
	\dot{\gamma}^m\partial^{\alpha}\partial_mf(0)=\gamma^{m}\partial^{\alpha}\partial_mh(0)
\end{gather}
where we once more used the explicit expression for $\gamma$. In the same way we deduce that 
\begin{gather}
	\label{6E13}
	\dot{\gamma}^m\partial^{\alpha}\partial_mh(0)=-\gamma^m\partial^{\alpha}\partial_mf(0)
\end{gather}
for all multi-indices $|\alpha|\leq n$. We can insert (\ref{6E12}) and (\ref{6E13}) into (\ref{6E7}) in order to deduce that the numerator equals minus the denominator so that (\ref{6E7}) simplifies to $\operatorname{ind}_p(X)=-(n+1)=-(m(p)+1)$ as desired.
\newline
Lastly, in order to justify that we may indeed reduce the fraction we have to argue why the denominator of the integrand in (\ref{6E7}) is non-zero. To see this, we note that if the denominator would be zero for some fixed $t\in [0,2\pi]$, then using repeatedly the identities (\ref{6E12}) and (\ref{6E13}) we would find that $$\gamma^{i_1}\text{..}\gamma^{i_k}\dot{\gamma}^{j_1}\text{..}\dot{\gamma}^{j_{n+1-k}}\partial_{i_1}\text{..}\partial_{i_k}\partial_{j_1}\text{..}\partial_{j_{n+1-k}}f(0)=0=\gamma^{i_1}\text{..}\gamma^{i_k}\dot{\gamma}^{j_1}\text{..}\dot{\gamma}^{j_{n+1-k}}\partial_{i_1}\text{..}\partial_{i_k}\partial_{j_1}\text{..}\partial_{j_{n+1-k}}h(0)$$ for every $0\leq k\leq n+1$ and consequently, since $\gamma(t)$ and $\dot{\gamma}(t)$ forms a basis of $\mathbb{R}^2$ for each $t\in [0,2\pi]$ we deduce that $\partial^{\widetilde{\alpha}}f(0)=0=\partial^{\widetilde{\alpha}}h(0)$ for all multi-indices $|\widetilde{\alpha}|\leq n+1$ contradicting the fact that $p$ is a zero of order $n$. Therefore, the result follows.
\end{proof}
We obtain the following immediate well-known corollary, see also \cite[Theorem 7]{PPS22} for a related, stronger result in the case of the $2$-torus which was obtained by a different approach.
\begin{cor}
	\label{6C5}
	Given an oriented, compact, smooth Riemannian $2$-manifold $(\Sigma,g)$ without boundary we define $\mathcal{H}(\Sigma):=\{\gamma\in \mathcal{V}(\Sigma) \mid \operatorname{div}_{\Sigma}(\gamma)=0=\operatorname{curl}_{\Sigma}(\gamma)\}$. Then the following holds
	\begin{enumerate}
		\item If $\Sigma$ is a surface of genus $h\in \mathbb{N}_0$, then every $\gamma\in \mathcal{H}(\Sigma)\setminus \{0\}$ has at most $(2h-2)$ zeros.
		\item If $\Sigma\cong S^2$, then $\mathcal{H}(\Sigma)=\{0\}$.
		\item If $\Sigma\cong T^2$, then every $\gamma\in \mathcal{H}(\Sigma)$ is either identically zero or no-where vanishing.
	\end{enumerate}
\end{cor}
\begin{proof}[Proof of \Cref{6C5}]
	If $\gamma\in \mathcal{H}(\Sigma)\setminus \{0\}$, then according to \Cref{6L4} $\gamma$ has only isolated zeros, since $\partial\Sigma=\emptyset$. Further, by means of the Poincar\'{e}-Hopf theorem, c.f. \cite[Chapter 6]{M65}, we must have $\sum_{p\in \{\gamma=0\}}\operatorname{ind}_p(\gamma)=\chi(\Sigma)=2-2h$ where $\chi(\Sigma)$ denotes the Euler-characteristic of $\Sigma$ and $h\in \mathbb{N}_0$ denotes the genus of $\Sigma$.
	
	We let $N$ denote the number of zeros of $\gamma$ and obtain from \Cref{6L4} and the above formula
	\begin{gather}
		\nonumber
		2-2h=\sum_{p\in \{\gamma=0\}}\operatorname{ind}_p(\gamma)=-\sum_{p\in \{\gamma=0\}}(m(p)+1)\leq -N\Rightarrow N\leq 2h-2.
	\end{gather}
	The remaining two statements follow from the fact that $S^2$ has genus $h=0$ and $T^2$ has genus $h=1$. 
\end{proof}
Lastly, we will require the following property of div-free fields.
\begin{lem}
	\label{6L6}
	Let $(\Sigma,g)$ be an oriented, connected (not necessarily compact) smooth Riemannian $2$-manifold without boundary and $j\in \mathcal{V}_0(\Sigma)$. Then for every proper periodic orbit $\gamma_p$ starting at some $p\in \Sigma$ of $j$ there exists some $\delta>0$ such that every orbit $\gamma_q$ of $j$ starting at some $q\in \Sigma$ with $\operatorname{dist}_g(\gamma_q,\gamma_p)<\delta$ is compactly ambient isotopic to $\gamma_p$, where $d_g$ denotes the standard  induced Riemannian distance.
\end{lem}
\begin{proof}[Proof of \Cref{6L6}]
	We start by defining the vector field $Y:=\mathcal{N}\times j$ where $\mathcal{N}$ denotes an arbitrary smooth unit normal field to $\Sigma$ (here we can think of $\Sigma$ being (not necessarily isometrically) embedded in $\mathbb{R}^3$ and the quantities involved being considered with respect to the pullback metric or equivalently $Y$ may be obtained from $j$ by identifying $j$ with a $1$-form via the musical isomorphism and applying the Hodge star operator to this $1$-form). We observe that $\operatorname{curl}_{\Sigma}(Y)=0$ and we let $\gamma$ denote the image of $\gamma_p$. Since $\gamma$ is a proper periodic orbit, we can find some open neighbourhood $U$ around $\gamma$ such that $j\neq 0$ in $U$ and consequently $Y\neq 0$ in $U$. Then we may consider the flowout of $\frac{Y}{|Y|^2}$ from $\gamma$, \cite[Theorem 9.20]{L12}, which gives us a diffeomorphism $\psi:(-\epsilon,\epsilon)\times \gamma\rightarrow U$ (after possibly shrinking $U$) with $\psi|_{\{0\}\times \gamma}=\operatorname{Id}$. We observe that $U\cong (0,1)\times S^1$ and that $\gamma$ is a generator of the first fundamental group of $U$. We can then verify that $\int_{\gamma}Y=0$ because $\gamma$ is tangent to $j$ and $j$ and $Y$ are perpendicular. Since $Y$ is curl-free and $\gamma$ a generator we conclude that $Y$ integrates to zero along any closed curve and is curl-free so that $Y$ admits a scalar potential within $U$, i.e. $Y=\nabla_{\Sigma}f$ for a suitable function $f\in C^{\infty}(U)$.
	
	Now a direct calculation shows that for fixed $p\in \gamma$, $\frac{d}{dt}f(\psi(t,p))=\nabla_{\Sigma} f\cdot \frac{Y}{|Y|^2}=1$ where we used that $Y=\nabla_{\Sigma}f$ and that by definition of the flowout $\psi(t,p)$ coincides with the integral curve of $\frac{Y}{|Y|^2}$ starting at $p$. From this one easily concludes that $\psi$ defines a diffeomorphism between the level sets of $f$. Further, since $\nabla_{\Sigma}f$ is normal to its level sets and $j\cdot \nabla_{\Sigma}f=0$, $j$ must be tangent to the level sets of $f$ so that if $d_{g}(\gamma_q,\gamma_p)$ is so small that $\gamma_q\cap U\neq \emptyset$ we conclude that $\gamma_q\subset U$ and is itself a periodic orbit.
	
	By the above arguments the image of $\gamma_q$ will coincide with some level set $\tau$ of $f$ (w.l.o.g. $\tau\geq 0$). We can then multiply $\frac{Y}{|Y|^2}$ by some bump function $\rho$ which is identically $1$ on $\psi([-\tau-\kappa,\tau+\kappa]\times \gamma)$ and identically zero on $\psi((-\epsilon,-\epsilon+\kappa)\times \gamma)\cup \psi((\epsilon-\kappa,\epsilon)\times \gamma)$ for small enough $\kappa>0$ such that $\rho\frac{Y}{|Y|^2}$ admits a globally defined flow $\Psi$ which coincides with the flowout $\psi$ upon restricting it to $[0,\tau]\times \gamma$. Then $\Psi:[0,\tau]\times \Sigma\rightarrow \Sigma$ gives rise to a continuous family of homeomorphisms such that $\gamma_p$ turns out to be ambient isotopic to $\Psi_{\tau}\circ \gamma_p$. Finally, $\gamma_q$ and $\Psi_{\tau}\circ \gamma_p$ have the same image and are oriented in the same way so that the smoothness of these curves implies that they are ambient isotopic, see the proof of \cite[Proposition A.1.13]{G20Diss} for a detailed account of this argument.
	
	Since the notion of ambient isotopy is an equivalence relation, c.f. \Cref{3R7}, we conclude that $\gamma_p$ and $\gamma_q$ are ambient isotopic as claimed.
\end{proof}
\subsection{Proof of \Cref{1T1}}
The following is the exact mathematical formulation of the theorem
\begin{thm}[Generic behaviour of toroidal currents]
	\label{6T7}
	Let $T^2\cong\Sigma\subset\mathbb{R}^3$ be a smooth toroidal surface and let $j\in \mathcal{V}_0(\Sigma)$. Then generically, c.f. \Cref{6L1}, we have the following dichotomy:
	\begin{enumerate}
		\item Either $j(p)\neq 0$ for all $p\in \Sigma$ and if this is the case, then $j$ is semi-rectifiable,
		\item or $j(p)=0$ for some $p\in \Sigma$. If this is the case, then the flow of $j$ has the following properties 
		\begin{enumerate}
			\item The flow of $j$ admits at least one centre- as well as saddle-point region, c.f. \Cref{3D1} and \Cref{3D2}, and the number of centre regions coincides with the number of saddle regions.
			\item Every singular point of $j$ is either a centre singularity or a saddle singularity.
			\item All but finitely many orbits of $j$ are periodic and the number of non-periodic orbits of $j$ coincides with the number of singular points of $j$.
			\item The $\omega$- and $\alpha$-limit set of each non-periodic orbit coincides with precisely one saddle point of $j$ respectively (the $\omega$- and $\alpha$-limit set may, but need not to, consist of the same point).
			\item Let $\Gamma\subset \Sigma$ denote the union of all saddle singularities and non-periodic orbits of $j$. Then $U:=\Sigma\setminus \Gamma$ is an open subset consisting of finitely many connected components $U_1,\dots,U_n$ for some $n\in \mathbb{N}$. \item If some $U_i$ contains a (necessarily centre) singularity then all orbits of $j$ starting in $U_i$ are contractible, i.e. they represent the trivial element of the first fundamental group of $\Sigma$.
			\item If $j(p)\neq 0$ for all $p\in U_i$ then any two orbits of $j$ starting in $U_i$ are ambient isotopic relative to $U_i$ within $\Sigma$.
		\end{enumerate}
	\end{enumerate} 
\end{thm}
\begin{rem}
	We emphasise that both, \Cref{6L1} as well as \Cref{6T7}, are valid more generally for toroidal surfaces $\Sigma$ equipped with an arbitrary smooth Riemannian metric $g$.
\end{rem}
\begin{proof}[Proof of \Cref{6T7}]
	We start by considering the situation $j(p)\neq 0$ for all $p\in \Sigma$. We can consider $Y:=\mathcal{N}\times j$ where $\mathcal{N}$ defines the outward unit normal. This defines a no-where vanishing curl-free vector field on $\Sigma$ with $Y\cdot j=0$ everywhere. The semi-rectifiability property is then a direct consequence of \cite[Theorem 9]{PPS22}.
	\newline
	We note that so far we did not make use of the genericity assumption on the vector fields $j$.
	\newline
	\newline
	We observe now that according to \Cref{6L1} it is enough to establish the dichotomy for vector fields $j\in \mathcal{V}_0(\Sigma)$ which have only non-degenerate zeros. It is clear that either $j(p)\neq 0$ for all $p\in \Sigma$ or that $j$ admits a zero in $\Sigma$. We are therefore left with establishing properties (a)-(g) under the assumption of the existence of a zero of $j$.
	\newline
	Suppose that $j(p)=0$ for some $p\in \Sigma$. Since all zeros of $j$ are non-degenerate, the index of any zero of $j$ is either $+1$ or $-1$, c.f. \cite[Chapter 6: Lemma 4, Lemma 5]{M65}. Further, since non-degenerate zeros are isolated and hence $j$ has only isolated zeros, it follows from the Poincar\'{e}-Hopf theorem, \cite[Chapter 6]{M65}, that the sum over all indices over all zeros of $j$ equals the Euler-characteristic of $\Sigma$ which equals $0$ since $\Sigma\cong T^2$. Since $j$ admits at least one zero by assumption, it must in fact have at least one zero $p_+$ of index $+1$ and one zero $p_-$ of index $-1$.
	\newline
	We argue now first that the zero of index $-1$ corresponds to a saddle region. To see this, we work in normal coordinates centred around $p_-$. We observe that the corresponding Jacobian at point $p_-$ admits $2$ (complex) eigenvalues $\lambda_1,\lambda_2\in \mathbb{C}$. Since $j$ is div-free it follows immediately from the fact that we work in normal coordinates that the Jacobian of $j$ at point $p_-$ has zero trace and consequently $\lambda:=\lambda_1=-\lambda_2$. Further, since all entries of the Jacobian are real, so must be its determinant, which is also given by $-\lambda^2=\det{Dj(p_-)}\in \mathbb{R}$ (where the Jacobian is computed in normal coordinates). Note further, that the Jacobian is non-singular at $p_-$ and therefore $\lambda\neq 0$. We conclude that either $\lambda$ is real and non-zero (which corresponds to the determinant being negative) or $\lambda$ is purely imaginary (and non-zero) which corresponds to the situation where the determinant is positive. Since the index at $p_-$ equals $-1$ and corresponds to the sign of the determinant of the Jacobian of $j$ at $p_-$, \cite[Chapter 6]{M65}, we conclude that $\lambda\in \mathbb{R}\setminus \{0\}$ and consequently both eigenvalues of $Dj(p_-)$ are real, non-zero and of opposite sign. As explained in \Cref{3R3} it follows now from the Hartman-Grobman theorem that $j$ admits a saddle region.
	\newline
	Now we will argue that the zero $p_+$ of index $+1$ corresponds to a centre region. Note that by the above arguments, the eigenvalues of the Jacobian of $j$ at $p_+$ are purely imaginary and hence the Hartman-Grobman theorem is no longer applicable. Instead, we fix any contractible neighbourhood $U$ around $p_+$ and define the vector field $w:=j\times \mathcal{N}$. Since $j$ is div-free, we conclude that $\operatorname{curl}_{\Sigma}(w)=0$ and further one verifies that $j=\mathcal{N}\times w$. By contractibility of $U$ we may express $w=\nabla_{\Sigma}f$ within $U$ (with $f(p_+)=0$). Further, it is clear that $\nabla_{\Sigma}f(p_+)=w(p_+)=0$ because $j(p_+)=0$ and one verifies easily by working in normal coordinates that $p_+$ is a non-degenerate zero of $w$ because it is a non-degenerate zero of $j$ and even more that the index of $w$ at $p_+$ equals $+1$. Consequently, $p_+$ is a non-degenerate critical point of $f$ in $U$ and the corresponding Hessian is either positive definite or negative definite. We can therefore apply the Morse lemma, \cite[Lemma 2.2]{Mil63}, and find a coordinate system $\mu$ around $p_+$ in which $f$ takes the form $f(x,y)=x^2+y^2$ and $\mu(p_+)=0$ where we without loss of generality assume that the Hessian of $f$ at $p_+$ is positive definite. We observe further by definition of $j$ and since $\nabla_{\Sigma}f=w$ that $df(j)=j(f)=j\cdot \nabla_{\Sigma}f=0$. In addition, it is immediate from the expression of $f$ in the given local coordinate system that the vector field $(-y,x)$ spans the kernel of the differential $df$ at every $(x,y)\in \mathbb{R}^2\setminus \{(0,0)\}$. We have already observed that $j(x,y)$ also lies in the kernel of $df$ for every $(x,y)$ and consequently (since the kernel is $1$-dimensional) we conclude that $j(x,y)=h(x,y)(-y,x)$ in our coordinate system for all $(x,y)\neq 0$ for a suitable non-vanishing scalar function $h$ (since $j$ does not vanish in a punctured neighbourhood of $p_+$). We note that the field lines of $(-y,x)$ are all circles not containing zero and that therefore on each circle $h$ is bounded away from zero. Hence, the field lines of $j$ also form circles around $(0,0)$. We can therefore let $W:=\mu^{-1}(B_r(0))$ be the preimage of a small enough Euclidean ball around $(0,0)$ and conclude that $W\setminus \{p_+\}$ consists of proper periodic orbits of $j$. This proves that the flow of $j$ admits a centre region.
	\newline
	We have therefore characterised the behaviour of the flow of $j$ around all of its singularities and shown that there is at least one saddle and one centre region.
	
	To prove the characterisation of the field lines of $j$ in statement (ii) we observe that locally, in normal coordinates centred around any fixed singular point, $j$ satisfies a Lojasiewicz-type inequality $|j(q)|\geq C|q|$ for some $C>0$ and all $q$ within this neighbourhood. If that were not the case there would be a sequence $q_n\rightarrow 0$ with $\frac{|j(q_n)|}{|q_n|}\rightarrow 0$ and since $j(0)=0$ in our coordinate system this would imply by definition of the differential that $Dj(0)\cdot \frac{q_n}{|q_n|}\rightarrow 0$. But by compactness of $S^1$ (after passing to a subsequence), $\frac{q_n}{|q_n|}\rightarrow q\neq 0$ and so $Dj(0)\cdot q=0$ for some $q\neq 0$, contradicting the invertibility of $Dj(0)$. It is then a direct consequence of \cite[Theorem 1]{CJL99} that all orbits of $j$ are proper periodic orbits except for its singular points and non-periodic orbits which connect singular points.
	
	Since points of index $+1$ correspond to centre regions and can therefore not be connected to because distinct orbits do not cross, we conclude that the non-periodic orbits all connect saddle singularities. It is now standard that the $\omega$-limit set is connected and since it is contained in the zero set of $j$ which consists of isolated points, we must have $\omega(q)=\{p\}$ whenever $q\in \Sigma$ is such that $\gamma_q$ is not periodic and where $p\in \Sigma$ is a suitable saddle singularity of $j$. Saddle singularities admit precisely two incoming and two outgoing field lines, c.f. \Cref{1F1}. Therefore each saddle singularity $p$ of $j$ will be contained in precisely two (distinct) $\omega$- and $\alpha$-limit sets of non-periodic trajectories. Hence there must be precisely $2\cdot \#\{p\in \Sigma\mid p\text{ is a saddle singularity}\}$ non-periodic trajectories. But as we have seen, according to the Poincar\'{e}-Hopf theorem, there are as many centre singularities as there are saddle singularities and each singularity is either of centre or of saddle type so that overall there are as many non-periodic orbits as there are singularities.
	
	We now let $\Gamma\subset \Sigma$ denote the union of all non-periodic orbits and saddle singularities of $j$ and we observe that by our previous arguments $\Gamma$ can equivalently be expressed as the union of the closures of the non-periodic orbits of $j$. Since there are finitely many such orbits, $\Gamma$ is closed as a union of finitely many closed sets and consequently $U=\Sigma\setminus \Gamma$ is open. Further, if $U$ were to have infinitely many connected components, we could extract a sequence $(p_n)_n\subset U$ such that each $p_n$ is contained in a distinct component and such that $p_n\rightarrow p$ for a suitable $p\in \Sigma$. By construction we must have $p\in \Gamma$. If $p$ is a saddle singularity, then according to \Cref{3D2}, see also \Cref{1F1}, the non-periodic orbits divide a suitable neighbourhood around $p$ in at most $4$ connected components so that infinitely many of the $p_n$ must be contained in the same component, which contradicts the choice of the $p_n$. The remaining case is that $p$ is contained in a non-periodic orbit $\gamma$ of $j$. But then, since the $\omega$- and $\alpha$-limit sets of the non-periodic field lines of $j$ consist only of saddle points we conclude that $\Gamma\setminus \gamma$ must have a positive distance to $p$ so that locally around $p$ at most two distinct connected components can meet and we obtain once again a contradiction. We conclude that $U$ has indeed finitely many connected components.
	
	Finally, fix any connected component $U_i$ of $U$, denote by $P_i:=\{p\in U_i\mid j(p)=0\}$ and set $W_i:=U_i\setminus P_i$. We observe that $P_i$ consists of at most finitely many points so that the $W_i$ are all connected, open sets and hence manifolds in their own right and by what we have already shown all orbits of $j$ starting at some $W_i$ are proper periodic orbits.
	
	We claim first that any two orbits within $W_i$ are ambient isotopic relative to $W_i$. To see this, fix $p\in W_i$ and define the set $S:=\{q\in W_i\mid \gamma_q\text{ is ambient isotopic to }\gamma_p\text{ relative to }W_i\}$ where $\gamma_p$ denotes the field line of $j$ starting at $p$. 
	Clearly $p\in S$ and so $S\neq\emptyset$. Further, if $q\in S$, then according to \Cref{6L6} there exists some $\delta>0$ such that $\gamma_x$ for all $x\in B_{\delta}(q)$ is compactly ambient isotopic to $\gamma_q$ within $W_i$ and so in particular ambient isotopic to $\gamma_q$ relative to $W_i\subset \Sigma$. Then by definition of $S$ and the transitivity of ambient isotopy relative to a set we conclude that $B_{\delta}(q)\subset S$ and so $S$ is open. Similarly, if $(q_n)_n\subset S$ converges to some $q\in W_i$, then by means of \Cref{6L6} for large enough $n$ it will be the case that $\gamma_{q_n}$ and $\gamma_q$ are compactly ambient isotopic within $W_i$ and again by transitivity we find $q\in S$, i.e. $S$ is closed. Overall, $S$ is a non-empty clopen subset of the connected set $W_i$ and thus $W_i=S$ which proves the claim. In particular, if $j(p)\neq 0$ in $U_i$, then $U_i=W_i$ and bullet point (g) follows
	
	Lastly, if $P_i\neq \emptyset$, i.e. $U_i$ contains a (centre) singularity of $j$, then since centre regions contain proper periodic orbits which are contractible within $\Sigma$, all orbits within $U_i$ must be contractible by the aforementioned ambient isotopy. More precisely all orbits starting at $P_i$ are constant orbits and thus contractible while all orbits starting in $W_i$ will be ambient isotopic to one of the contractible orbits surrounding a centre singularity and hence itself be contractible.
\end{proof}
\subsection{Proof of \Cref{1T2}}
We state now the precise mathematical formulation of the theorem
\begin{thm}[Generic behaviour of cylindrical currents]
	\label{6T8}
	Let $S^1\times [0,1]\cong \Sigma\subset\mathbb{R}^3$ be a smooth cylindrical surface and let $j\in \mathcal{V}^{\operatorname{OT}}_0(\Sigma)$. Then generically, c.f. \Cref{6L2}, the flow of $j$ has the following properties:
	\begin{enumerate}
		\item The flow of $j$ admits at least one centre and one saddle region and the number of centre regions coincides with the number of saddle regions.
		\item Every singular point of $j$ is either a centre singularity or a saddle singularity.
		\item All but finitely many orbits of $j$ are periodic and the number of non-periodic orbits of $j$ coincides with the number of singularities of $j$.
		\item The $\omega$- and $\alpha$-limit set of each non-periodic orbit coincides with precisely one saddle point of $j$ respectively (the $\omega$- and $\alpha$-limit set may, but need not to, consist of the same point).
		\item Let $\Gamma\subset \Sigma$ denote the union of all saddle singularities and non-periodic orbits of $j$. Then $U:=\Sigma\setminus \Gamma$ is an open subset consisting of finitely many connected components $U_1,\dots,U_n$ for some $n\in \mathbb{N}$.
		\item If $U_i$ contains a singularity, then all orbits of $j$ starting in $U_i$ are contractible.
		\item If $U_i\cap \partial\Sigma=\emptyset$ and $U_i$ does not contain any singularities of $j$, then any two orbits of $j$ starting in $U_i$ are ambient isotopic relative to $U_i$ within $\operatorname{int}(\Sigma)$.
		\item The connected components $U_i$ with $U_i\cap \partial\Sigma\neq\emptyset$ consist of poloidal orbits of $j$, i.e. of proper periodic orbits which represent non-trivial elements of the first fundamental group of $\Sigma$.
	\end{enumerate}
\end{thm}
\begin{proof}[Proof of \Cref{6T8}]
	According to \Cref{6L2} we only need to prove the statement for vector fields $j\in \mathcal{V}_0^{\operatorname{OT}}(\Sigma)$ with non-degenerate zeros within $\operatorname{int}(\Sigma)$.
	\newline
	Our goal now will be to argue that $j$ has at least one zero of index $-1$ in $\operatorname{int}(\Sigma)$ and one zero of index $+1$ in $\operatorname{int}(\Sigma)$. Then the arguments from the proof of \Cref{6T7} regarding the existence of saddle and centre regions will apply verbatim and show that zeros of index $+1$ correspond to centre regions while zeros of index $-1$ correspond to saddle regions.
	\newline
	Just like in the proof of \Cref{6L2} we can find some $\psi\in C^{\infty}(\Sigma)$ such that $\psi|_{\partial\Sigma}$ is locally constant and  $j=\mathcal{N}\times \nabla_{\Sigma}\psi$. We observe further that the fact that $j\neq 0$ on $\partial\Sigma$ and that the field lines of $j$ are oppositely oriented on the boundary components imply that $\nabla_{\Sigma}\psi\neq 0$ on $\partial\Sigma$, $\nabla_{\Sigma}\psi\perp \partial\Sigma$ and that $\nabla_{\Sigma}\psi$ is either outward pointing along both boundary circles or is inward pointing along both boundary circles, see \Cref{6F1}.
	
	\begingroup\centering
	\begin{figure}[H]
		\hspace{4.5cm}\includegraphics[width=0.35\textwidth, keepaspectratio]{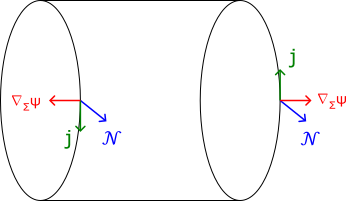}
		\caption{The outward unit normal $\mathcal{N}$ in blue. The outward pointing $\nabla_{\Sigma}\psi$ in red and $j=\mathcal{N}\times \nabla_{\Sigma}\psi$ in green.} 
		\label{6F1}
	\end{figure}
	\endgroup
	
	It then first of all follows from the Poincar\'{e}-Hopf theorem for manifolds with boundary, c.f. \cite[Chapter 6]{M65}, and the fact that cylindrical surfaces have zero Euler characteristic that the sum of all indices of the zeros of $j$ equals zero. Since $j$ has only non-degenerate zeros, we conclude just like in the proof of \Cref{6T7} that if $j$ has at least one zero, then it must in fact have at least one zero $p_-$ of index $-1$ and at least one zero $p_+$ of index $+1$. Even more, the number of zeros of index $+1$ must coincide with the number of zeros of index $-1$.
	
	To see that $j$ has at least one zero (we assume without loss of generality that $\nabla_{\Sigma}\psi$ is everywhere outward pointing) we make use of the fact that the values of $\psi$ are (strictly) increasing as we approach the boundary circles of $\partial\Sigma$ since $\nabla_{\Sigma}\psi\perp \partial\Sigma$ is non-vanishing and outward pointing. Therefore, letting $\partial\Sigma=C_1\cup C_2$ be the decomposition into the two boundary circles, there exist points $q_1,q_2\in \operatorname{int}(\Sigma)$ such that $\psi(q_1)<\min_{q\in C_1}\psi(q)$ and $\psi(q_2)<\min_{q\in C_2}\psi(q)$ and consequently there is some $\widetilde{q}\in \operatorname{int}(\Sigma)$ with $\psi(\widetilde{q})<\min_{q\in \partial\Sigma}\psi(q)$ and hence (by compactness) $\psi$ must admit a global minimum within $\operatorname{int}(\Sigma)$ and hence admits at least one critical point $x$. This implies $\nabla_{\Sigma}\psi(x)=0$ and thus $j(x)=\mathcal{N}(x)\times \nabla_{\Sigma}\psi(x)=0$ admits at least one zero within $\operatorname{int}(\Sigma)$. As explained previously, this implies the presence of at least one zero of index $-1$ and one zero of index $+1$ which correspond to a saddle and centre region respectively.
	
	We will now argue that all but finitely many orbits are proper periodic orbits and that the remaining orbits are the equilibria and non-periodic orbits connecting saddles.
	
	First, we note that since $\nabla_{\Sigma}\psi$ is everywhere outward pointing along $\partial\Sigma$ the (non-empty) level sets $\{q\in \operatorname{int}(\Sigma)\mid \psi(q)=s\}$ for any fixed $s\in \mathbb{R}$ have a positive distance to $\partial\Sigma$. This implies that all regular level sets of $\psi$ within $\operatorname{int}(\Sigma)$ are smooth, compact $1$-manifolds without boundary, i.e. they are a finite union of circles, see \cite[Appendix]{M65}. Further, $j\cdot \nabla_{\Sigma}\psi=0$ and so $j$ is tangent to the level sets of $\psi$ and consequently all orbits of $j$ starting at a regular level set of $\psi$ are proper periodic orbits. 
	Clearly also the two boundary circles are proper periodic orbits of $j$. Generic fields $j$ have only finitely many singular points (all contained in $\operatorname{int}(\Sigma)$) $p_1,\dots,p_N\in \operatorname{int}(\Sigma)$ and so there are finitely many critical level sets. We can now fix any such critical level set and observe that if we remove the corresponding critical points from this level set, the remaining level set will be a smoothly embedded $1$-manifold without boundary and so its connected components will either correspond to proper periodic orbits or to orbits which are diffeomorphic to $\mathbb{R}$ and the orbits diffeomorphic to $\mathbb{R}$ are precisely the non-periodic orbits.
	
	Therefore, if $\gamma_q$ is a non-periodic orbit, then $q$ is contained in a connected component of a critical level set of $\psi$ from which the critical points were removed and this connected component is diffeomorphic to $\mathbb{R}$. It is again standard that the $\omega$-limit set is non-empty and connected. If there were to exist some $p\in \omega(q)$ with $j(p)\neq 0$, i.e. $\nabla_{\Sigma}\psi(p)\neq 0$, then $p$ would be an element of the image of $\gamma_q$ because the critical level set with the critical points removed is a $1$-manifold and consequently we would reach $p$ in finite time, i.e. $\gamma_q$ would be a proper periodic orbit, contradicting the fact that the image of $\gamma_q$ is diffeomorphic to $\mathbb{R}$. Consequently $\omega(q)\subset \{p\in \operatorname{int}(\Sigma)\mid j(p)=0\}$ and since $\omega(q)$ is connected and the zero set isolated we conclude that $\omega(q)$ is a singleton.
	
	From here on we can argue identically as in the proof of the classification of the finitely many non-periodic orbits in statement (ii) of \Cref{6T7} in order to conclude that there are exactly as many non-periodic orbits as there are singular points of $j$ and that the $\omega$- and $\alpha$-limit set of any such orbit coincide with some saddle singularity of $j$ respectively.
	
	To conclude the proof we can consider $U:=\Sigma\setminus \Gamma$ and just like in \Cref{6T7} it follows that $U$ has at most finitely many connected components. At most two of these components will intersect $\partial\Sigma$ and upon removing the boundary circles from these components we can argue identically as in the proof of \Cref{6T7} that in each such component all orbits are contractible if $U_i$ contains a singularity or otherwise all components within $U_i\setminus\partial \Sigma$ are ambient isotopic within $\operatorname{int}(\Sigma)$ relative to $U_i\setminus \partial\Sigma$. Lastly, if $\partial\Sigma\cap U_i\neq \emptyset$, then we can consider the inward flow of $Y:=\frac{\nabla_{\Sigma}(-\psi)}{|\nabla_{\Sigma}\psi|^2}$ which is inward pointing along the boundary and provides a diffeomorphism between the boundary circles and the (interior) level sets of $\psi$. We recall that $j=\mathcal{N}\times \nabla_{\Sigma}\psi$ and so the field lines of $j$ near the boundary circles coincide with the level sets of $\psi$. Accordingly, if $U_i$ contains a boundary circle, then there are field lines of $j$ within $U_i\setminus \partial\Sigma$ which are homotopic to a non-contractible boundary circle and hence $U_i$ must consist of poloidal orbits alone.
\end{proof}
\subsection{Proof of \Cref{1T3}}
In this subsection we prove the following result, where we recall that $\mathcal{H}_N(\Sigma)=\{j\in \mathcal{V}(\Sigma)\mid \operatorname{div}_{\Sigma}(j)=0=\operatorname{curl}_{\Sigma}(j)\text{, }j\parallel \partial\Sigma\}$.
\begin{thm}
	\label{6T9}
	Let $[0,1]\times S^1\cong \Sigma\subset \mathbb{R}^3$ be a smooth cylindrical surface and let $j\in \mathcal{H}_N(\Sigma)$. Then
	\begin{enumerate}
		\item $j$ is either identically zero or is otherwise no-where vanishing.
		\item If $j$ is not identically zero, then all of its orbits are proper poloidal orbits, i.e. they represent non-trivial elements of the first fundamental group of $\Sigma$.
	\end{enumerate}
\end{thm}
\begin{proof}[Proof of \cref{6T9}]
	We assume that $j$ is not identically zero and argue first that all zeros of $j$ are contained in  $\operatorname{int}(\Sigma)$. First, according to (\ref{1E2})-(\ref{1E5}) we write $j=I\mathcal{N}\times \nabla_{\Sigma}\psi$ where $I\in \mathbb{R}$ is a constant and $\psi$ is a solution of
	\begin{gather}
		\label{6E14}
		\Delta \psi=0\text{ in }\Sigma\text{, }\psi|_{c_1}=0,\psi|_{c_2}=1
	\end{gather}
	where $\partial\Sigma=c_1\cup c_2$ is the decomposition of $\partial\Sigma$ into its two boundary circles. We note that the properties stated in \Cref{6T9} are preserved by multiplications by non-zero constants of $j$ and so we may assume that $I=1$.
	\newline
	We first argue that $j(p)\neq 0$ for all $p\in \partial\Sigma$. To this end we note that according to the maximum principle, \cite[Chapter 6.4.2 Theorem 3]{Evans10}, since $\psi$ is non-constant, $\psi$ achieves its minimum and maximum only on its boundary, so that $\min_{p\in \Sigma}\psi(p)=0$, $\max_{p\in \Sigma}\psi(p)=1$. It then follows from Hopf's lemma, \cite[Chapter 6.4.2]{Evans10}, that $n\cdot \nabla_{\Sigma}\psi\neq 0$ on all of $\partial\Sigma$ (because every boundary point is a global extremum), where $n$ denotes a unit normal to $\partial\Sigma$ within $\Sigma$, i.e. $n\perp \partial\Sigma$ and $n\cdot \mathcal{N}=0$ where $\mathcal{N}$ denotes a smooth unit normal to the surface $\Sigma$. Consequently $\nabla_{\Sigma}\psi\neq 0$ on $\partial\Sigma$ and thus $j=\mathcal{N}\times \nabla_{\Sigma}\psi\neq 0$ on $\partial\Sigma$ as claimed.
	\newline
	\newline
	Now we will argue that $j$ has at most finitely many zeros. According to \Cref{6L4} the set $N:=\{q\in \operatorname{int}(\Sigma)\mid \nabla_{\Sigma}\psi(q)=0\}$ consists of at most countably many points which do not accumulate within $\operatorname{int}(\Sigma)$. Hence, if $N$ would be infinite, then by compactness of $\Sigma$ we could extract a sequence $(p_n)_n\subset N$ converging to some $p\in \Sigma$. Since $N$ does not accumulate within $\operatorname{int}(\Sigma)$, we find $p\in \partial\Sigma$ and by continuity $\nabla_{\Sigma}\psi(p)=0$, i.e. $j(p)=0$, contradicting the first part of our proof. Hence, indeed $j$ has at most finitely many zeros, all contained within $\operatorname{int}(\Sigma)$.
	\newline
	\newline
	We will now show that $j$ does not have any zeros at all. We assume for a contradiction that $N\neq \emptyset$. We then first employ \Cref{6L4} to deduce that
	\begin{gather}
		\label{6E15}
		\sum_{p\in N}\operatorname{ind}_p(\nabla_{\Sigma}\psi)=-\sum_{p\in N}(m(p)+1)\leq -\sum_{p\in N}1\leq -1<0
	\end{gather}
	where we used that $\psi$ is harmonic, recall (\ref{6E14}). We note that since $\psi$ achieves its global minimum only on $c_1$ and its global maximum only on $c_2$, $\nabla_{\Sigma}\psi$ is inward pointing along $c_1$ and outward pointing along $c_2$ so that we cannot utilise directly the Poincar\'{e}-Hopf theorem, \cite[Chapter 6]{M65}. Instead, we can embed $[0,1]\times S^1\cong \Sigma$ into $\mathbb{R}^2$, more precisely we can map it onto a standard annulus $A:=D_2\setminus \operatorname{int}(D_1)$ where $D_r$ denotes the closed $2$-dimensional Euclidean disc centred at $0\in \mathbb{R}^2$. We can then view $\psi$ as a function on $A$ (which with a slight abuse of notation we denote in the same way) and extend it to a smooth function $\widetilde{\psi}$ on $D_2$ in such a way that $\nabla_g\widetilde{\psi}$ (where $g$ is an arbitrary extension of the pull-back metric stemming from the embedding $\Sigma\hookrightarrow \mathbb{R}^2$) has only isolated zeros. This can be achieved similar in spirit as in the proof of \Cref{6L2} by "gluing" an arbitrary extension of $\psi$ together with a suitable Morse function. We can now orient the boundary circles with respect to the outward normal direction, i.e. we orient $\partial D_2$ counter-clockwise and $\partial D_1$ clock-wise. We express $\nabla_g\psi=f \partial_1+h\partial_2$ in standard Euclidean coordinates and then finally make use of the fact, c.f. \cite[Corollary 7.7]{Ful95}, that $\frac{1}{2\pi}\int_{\partial D_2}\frac{f dh-h df}{f^2+h^2}=\sum_{p\in \{q\in D_2\mid \nabla \widetilde{\psi}(q)=0\}}\operatorname{ind}_p(\nabla_g\widetilde{\psi})$ and similarly, since $\partial D_1$ is oppositely oriented, $\frac{1}{2\pi}\int_{\partial D_1}\frac{f dh-h df}{f^2+h^2}=-\sum_{p\in \{q\in D_1\mid \nabla \widetilde{\psi}(q)=0\}}\operatorname{ind}_p(\nabla_g\widetilde{\psi})$ so that we obtain
	\begin{gather}
		\label{6E16}
		\frac{1}{2\pi}\int_{\partial A}\frac{f dh-h df}{f^2+h^2}=\sum_{p\in N}\operatorname{ind}_p(\nabla_{\Sigma}\psi).
	\end{gather}
	We will now argue that $\int_{\partial A}\frac{f dh-h df}{f^2+h^2}=0$. To see this we observe that $\nabla_{\Sigma}\psi$ is outward pointing along $\partial D_2$ and hence if we let $\gamma_2(t):=2(\cos(t),\sin(t))$ be a parametrisation of $\partial D_2$ we may write $\nabla_{\Sigma}\psi(\gamma_2(t))=\rho(t)\frac{\gamma_2(t)}{2}+\lambda(t)\frac{\dot{\gamma}_2(t)}{2}$ for suitable, smooth, $2\pi$-periodic functions $\rho,\lambda$ where $\rho(t)>0$ for all $t$. Consequently we find $f(\gamma_2(t))=\rho(t)\cos(t)-\lambda(t)\sin(t)$, $h(\gamma_2(t))=\rho(t)\sin(t)+\lambda(t)\cos(t)$ and compute
	\begin{gather}
		\nonumber
		\int_{\partial D_2}\frac{f dh-hdf}{f^2+h^2}=\int_0^{2\pi}\frac{f \dot{h}-h\dot{f}}{f^2+h^2}dt=\int_0^{2\pi}\frac{\rho^2+\lambda^2+\rho\dot{\lambda}-\dot{\rho}\lambda}{\rho^2+\lambda^2}dt=2\pi+\int_0^{2\pi}\frac{\rho\dot{\lambda}-\dot{\rho}\lambda}{\rho^2+\lambda^2}dt.
	\end{gather}
	We then observe that since $\rho(t)>0$ for all $t$, we may consider the well-defined function $k(t):=\arctan\left(\frac{\lambda(t)}{\rho(t)}\right)$ and a direct computation yields $\dot{k}(t)=\frac{\rho\dot{\lambda}-\dot{\rho}\lambda}{\rho^2+\lambda^2}$ so that $\int_0^{2\pi}\frac{\rho\dot{\lambda}-\dot{\rho}\lambda}{\rho^2+\lambda^2}dt=0$ since $k(t)$ is $2\pi$-periodic. We conclude $\frac{1}{2\pi}\int_{\partial D_2}\frac{f dh-h df}{f^2+h^2}=1$ and in the same way we conclude $\frac{1}{2\pi}\int_{\partial D_2}\frac{f dh-h df}{f^2+h^2}=-1$ where the minus stems from the fact that we orient $\partial D_1$ in clock-wise direction. Consequently $\frac{1}{2\pi}\int_{\partial A}\frac{f dh-h df}{f^2+h^2}=1-1=0$ as claimed. We can therefore combine (\ref{6E16}) and (\ref{6E15}) in order to deduce
	\begin{gather}
		\nonumber
		0=\frac{1}{2\pi}\int_{\partial A}\frac{f dh-h df}{f^2+h^2}=\sum_{p\in N}\operatorname{ind}_p(\nabla_{\Sigma}\psi)=-\sum_{p\in N}(m(p)+1)<0
	\end{gather}
	which is a contradiction. We conclude that $\{q\in \operatorname{int}(\Sigma)\mid \nabla_{\Sigma}\psi(q)=0\}=N=\emptyset$ and thus altogether $\nabla_{\Sigma}\psi$ is no-where vanishing in $\Sigma$ and thus so is $j=\mathcal{N}\times \nabla_{\Sigma}\psi$.
	
	We refer the reader at this point also to \cite[Theorem 1.1, Theorem 2.1 \& Theorem 2.2]{AM92} for related results regarding the multiplicities of critical points of (Euclidean) harmonic functions in the plane as well as of functions satisfying more general elliptic equations.
	\newline
	\newline
	In order to conclude the proof we are left with arguing that $j$ has only poloidal periodic orbits. First, since $j\parallel \partial\Sigma$ and $j\neq 0$ on $\Sigma$ we conclude that the two boundary circles (which are both poloidal) are field lines of $j$. We have also seen in the beginning of this proof that $\psi$ achieves its global minimum precisely on the boundary circle $c_1$ and its global maximum precisely on the boundary circle $c_2$. This implies that $0<\psi<1$ within $\operatorname{int}(\Sigma)$ and that each level set $\{\psi=s\}$ for any fixed $0<s<1$ is compact and a regular level set since $\nabla_{\Sigma}\psi\neq 0$ on $\Sigma$. We conclude that each level set $\{\psi=s\}$ is a finite union of smoothly embedded circles, see \cite[Appendix]{M65} for a classification of $1$-manifolds. We further observe that $j\cdot \nabla_{\Sigma}\psi=0$ since $j=\mathcal{N}\times \nabla_{\Sigma}\psi$. Therefore, $j$ is tangent to and no-where vanishing on the level sets of $\psi$ so that all orbits of $j$ must be non-constant and periodic. According to \Cref{6L3} none of these periodic orbits can be contractible and consequently all orbits of $j$ must be poloidal, i.e. represent non-trivial elements of the first fundamental group.
	
	The fact that the flow of $j$ does not admit any saddle or centre regions follows immediately from the fact that $j\neq 0$ on all of $\Sigma$.
\end{proof}
\subsection{Proof of \Cref{1T5}}
We will prove the following result.
\begin{thm}
	\label{6T10}
	Let $S^1\times [0,1]\cong \Sigma\subset\mathbb{R}^3$ be a smooth cylindrical surface and $j\in \mathcal{V}_0(\Sigma)$ with $j\parallel \partial\Sigma$. Then there exists some $S\subset \Sigma$ of zero area such that every orbit of $j$ starting at some $p\in \Sigma\setminus S$ is periodic.
	
	In addition, if $\gamma_p$ is a non-periodic orbit of $j$ starting at some $p\in \Sigma$, then $\omega(p)\subset \{q\in \Sigma\mid j(q)=0\}$. In particular, if $j(p)\neq 0$ for all $p\in \Sigma$, then all orbits of $j$ are proper periodic orbits and they all represent non trivial elements of the first fundamental group of $\Sigma$.
\end{thm}
\begin{proof}[Proof of \Cref{6T10}]
	The first part of the proof is a modification of the proof of the Poincar\'{e}-Bendixson theorem given in \cite[Theorem 14.1.1]{HasKa95}.
	\newline
	We first note that according to Poincar\'{e}'s recurrence theorem, c.f. \Cref{1T4}, we can find a set $S\subset \Sigma$ of zero area such that every $p\in \Sigma\setminus S$ is a (forward) recurrent point, i.e. there exists a sequence $t_k\equiv t_k(p)\rightarrow \infty$ such that $\gamma_p(t_k)\rightarrow p$, where $\gamma_p$ denotes the field line of $j$ starting at $p$. We claim that any such orbit is already periodic. Since the boundary is of zero area we may assume that $p\in \operatorname{int}(\Sigma)$. If $j(p)=0$, there is nothing to show so that we also assume from now on that $j(p)\neq 0$ and $p\in \omega(p)$, i.e. the point $p$ is forward recurrent (note that in contrast to the statement in \cite[Theorem 14.1.1]{HasKa95} we do not know a priori whether the whole orbit $\gamma_p(\mathbb{R})$ is forward recurrent). Now according to \cite[Theorem 9.22]{L12} we can find a coordinate system $\mu:U\rightarrow \mathbb{R}^2$ around $p$ such that $j$ has the representation $\partial_1$ within the region $U$. We can then consider the vector field $Y:= j\times \mathcal{N}$ where $\mathcal{N}$ denotes the outward unit normal to $\Sigma$. Since $j\neq 0$ within $U$, so is $Y$ and we may consider the (locally defined) field line $\gamma^Y_p$ of $Y$ starting at $p$ which is everywhere transversal to $j$. After possibly shrinking $U$ we may assume that $\mu(U)=D_r$ (the $2$-dimensional disc of radius $r$ centred at $0$) and that $\gamma^Y_p$ intersects each line parallel to the $x$-axis within $D_r$ exactly once, see \Cref{6F2}.
	
	\begingroup\centering
	\begin{figure}[h]
		\hspace{3.5cm}\includegraphics[width=0.5\textwidth, keepaspectratio]{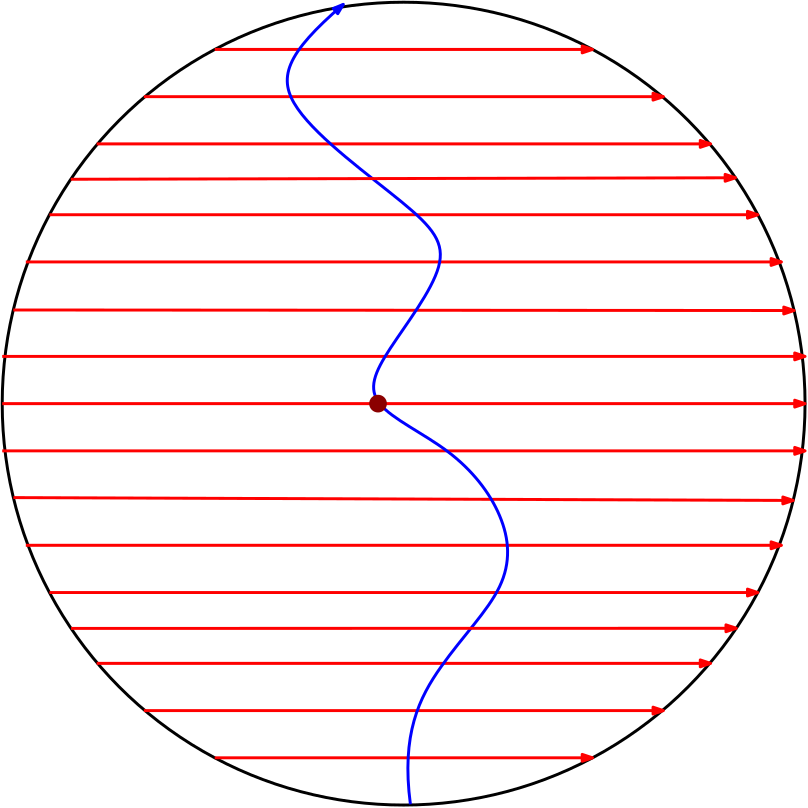}
		\caption{Depicted the disc $D_r$. The point $p$ marked as a thick dot in the centre. The field lines of $j$ (in red) correspond to straight lines passing from left to right. The field line $\gamma^Y_p$ (in blue) is transversal to $j$, i.e. never becomes parallel to the x-axis and thus moves at every instant upwards (or at every instant downwards).}
		\label{6F2}
	\end{figure}
	\endgroup
	
	Since $p$ is (forward) recurrent there will be a smallest $\tau>0$ such that $\gamma_p(\tau)\cap \gamma^Y_p\cap D_r\neq \emptyset$. We can then consider the path concatenation of $\gamma_p[0,\tau]$ and $\gamma^Y_p[0,\tau_Y]$ (possibly inverting the orientation if necessary) to obtain a closed (piecewise smooth) curve $\gamma$, see \Cref{6F3}, where $\tau_Y$ is such that $\gamma^Y_p(\tau_Y)=\gamma_p(\tau)$.
	
	\begingroup\centering
	\begin{figure}[h]
		\hspace{3.5cm}\includegraphics[width=0.5\textwidth, keepaspectratio]{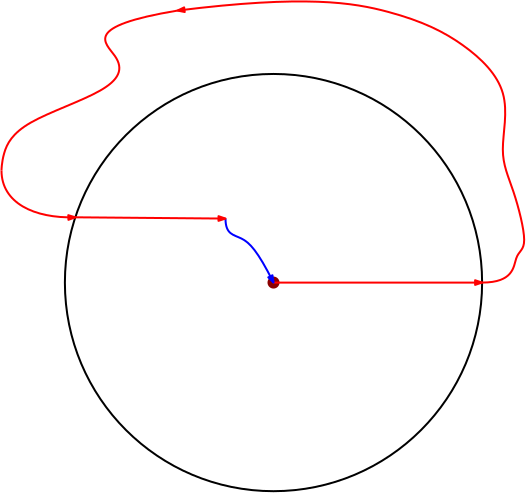}
		\caption{Depicted the disc $D_r$ with the point $p$ marked as a thick dot in the centre. The field line $\gamma_p$ of $j$ (in red) is a straight line within $D_r$ but may curve outside of $D_r$. Eventually the field line $\gamma_p$ enters once again the disc $D_r$ and will intersect the field line $\gamma^Y_p$ (in blue) which corresponds to the curve segment within $D_r$ which is not parallel to the $x$-axis. Together these two curves form the closed curve $\gamma$.}
		\label{6F3}
	\end{figure}
	\endgroup
	
	It then follows from the Hodge-decomposition theorem, \cite[Theorem 2.4.2 \& Theorem 2.4.8]{S95}, and (\ref{1E2})-(\ref{1E5}) that $j=\mathcal{N}\times \nabla_{\Sigma}f$ for a suitable function $f\in C^{\infty}(\Sigma)$ which is locally constant on $\partial\Sigma$, where we used that $j$ is div-free and $j\parallel \partial\Sigma$. Then on the one hand, since $\gamma$ is a closed curve, we find
	\begin{gather}
		\label{6E17}
		\int_{\gamma}df=0.
	\end{gather}
	On the other hand $\gamma=\gamma_p[0,\tau]\oplus \gamma^Y_p[0,\tau_Y]$ and so we also find
	\begin{gather}
		\nonumber
		\int_{\gamma}df=\int_{\gamma_p[0,\tau]}df+\int_{\gamma^Y_p[0,\tau_Y]}df=\int_{\gamma^Y_p[0,\tau_Y]}df=\int_0^{\tau_Y}|Y\left(\gamma^Y_p(s)\right)|^2ds
	\end{gather}
	where we used that $\dot{\gamma}_p$ is tangent to $j$, that $j$ and $\nabla_{\Sigma}f$ are pointwise orthogonal to each other and that $Y=j\times \mathcal{N}=\nabla_{\Sigma}f$. It follows therefore from (\ref{6E17}) that $0=\int_0^{\tau_Y}|Y(\gamma^Y_p(s))|^2ds$ and hence, since $Y\neq 0$ within $D_r$, $\tau_Y=0$. But that means by definition of $\tau$ that $\gamma_p(\tau)=\gamma^Y_p(\tau_Y)=\gamma^Y_p(0)=p$ and hence $\gamma_p$ is periodic since $\tau>0$.
	\newline
	\newline
	We will now lastly argue that if $q\in \Sigma$ is such that $\gamma_q$ (the orbit of $j$ starting at $q$) is non-periodic, then $\omega(q)\subset \{p\in \Sigma\mid j(p)=0\}$.
	
	We suppose for a contradiction that there exists some $p\in \omega(q)$ with $j(p)\neq 0$. We can then as before go into a coordinate chart around $p$ such that $j=\partial_1$ within this chart \cite[Theorem 9.22 \& Theorem 9.35]{L12} and as before we may consider $Y=j\times \mathcal{N}$ (selecting the unit normal $\mathcal{N}$ such that $Y$ is inward pointing in case that $p\in \partial\Sigma$). We can similarly consider the field line $\gamma_p^Y$ of $Y$ starting at $p$ (which will be defined on the interval $[0,\epsilon)$ for some $\epsilon>0$ if $p\in \partial\Sigma$) and since $p$ is recurrent there will be a smallest time $\tau_1>0$ such that $\gamma_q$ intersects $\gamma_p^Y\cap D_r$ and a smallest time $\tau_2>\tau_1$ such that $\gamma_q$ intersects once more $\gamma_p^Y\cap D_r$. We can then consider the closed curve $\widehat{\gamma}$ consisting of the curve segment $\gamma_q([\tau_1,\tau_2])$ and the part of $\gamma^Y_p$ connecting these two points. This gives us a closed curve and we can proceed to argue identically as in (\ref{6E17}) that the two intersection points of $\gamma_q$ at times $\tau_1$ and $\tau_2$ with $\gamma^Y_p$ must coincide and hence $\gamma_q$ is a proper periodic orbit, contradicting our assumption that $\gamma_q$ is non-periodic.
	
	Consequently, if $j\neq 0$ in $\Sigma$, then $j$ cannot have non-periodic orbits since every orbit has a non-empty $\omega$-limit set. If any such orbit would be contractible, it would bound a disc and as argued on several occasions before that would imply the existence of a zero in $\Sigma$, which is absurd.
	\end{proof}

\addsec{Acknowledgements}
This work has been supported by the Inria AEX StellaCage. The research was supported in part by the MIUR Excellence Department Project awarded to Dipartimento di Matematica, Universit`a di Genova, CUP D33C23001110001.
We would like to thank Mario Sigalotti and Ugo Boscain for useful discussions,
as well as the Renaissance Fusion SCIENCE team.

\appendix
\section{Generic flows on cylindrical surfaces with more general boundary conditions}
\label{A}
\begin{lem}
	\label{ExtraL1}
	Let $S^1\times [0,1]\cong \Sigma\subset\mathbb{R}^3$ be a smooth cylindrical surface and consider the set $\mathcal{V}^{\operatorname{T}}_0(\Sigma):=\{j\in \mathcal{V}_0(\Sigma)\mid j\parallel\partial\Sigma\text{ and }j\neq 0\text{ on }\partial\Sigma\}\subset \mathcal{V}(\Sigma)$. Then the set
	\begin{gather}
		\nonumber
		\{j\in \mathcal{V}^{\operatorname{T}}_0(\Sigma)\mid j\text{ has only non-degenerate zeros}\}
	\end{gather}
	is a generic subset of $\mathcal{V}^{\operatorname{T}}_0(\Sigma)$.
\end{lem}
\begin{proof}[Proof of \Cref{ExtraL1}]
	The proof of \Cref{6L2} applies mutatis mutandis.
\end{proof}
\begin{thm}
	\label{ExtraT2}
	Let $S^1\times [0,1]\cong \Sigma\subset \mathbb{R}^3$ be a smooth cylindrical surface and let $j\in \mathcal{V}^{\operatorname{T}}_0(\Sigma)$. Then generically, c.f. \Cref{ExtraL1}, the flow of $j$ satisfies the following dichotomy
	\begin{enumerate}
		\item $j(p)\neq 0$ for all $p\in \Sigma$. If this is the case, then all field lines of $j$ are proper periodic orbits which are all poloidal, i.e. represent a non-trivial element of the first fundamental group of $\Sigma$.
		\item $j(p)=0$ for some $p\in \Sigma$. If this is the case then the flow of $j$ satisfies properties (i)-(viii) listed in \Cref{6T8}
	\end{enumerate}
\end{thm}
\begin{proof}[Proof of \Cref{ExtraT2}]
	If $j(p)\neq 0$ for all $p\in \Sigma$ then the conclusion about the flow follows from \Cref{6T10}.
	
	If, on the other hand, $j(p)=0$ for some $p\in \Sigma$, then either $j$ is oppositely oriented on the boundary circles, so that the statement is simply a consequence of \Cref{6T8}. Otherwise, due to the non-degeneracy of the zeros, $j$ will admit only finitely many zeros and we can argue identically as in the proof of \Cref{6T9} that $\sum_{p\in \{j=0\}}\operatorname{ind}_p(j)=0$. With this knowledge and knowing that $j$ admits at least one zero by assumption, the remaining arguments of \Cref{6T8} apply and the theorem follows.
\end{proof}
\bibliographystyle{plainnat}
\bibliography{mybibfileNOHYPERLINK}
\footnotesize

\end{document}